\documentclass[11pt, a4paper, reqno,oneside]{amsart}
\usepackage{color}

\usepackage{amssymb}

\usepackage{bm}
\usepackage{bbm}

\usepackage{mathcomp}
\usepackage{dsfont}
\usepackage[toc,page]{appendix}
\usepackage{graphicx} 
\usepackage{subfigure}
\usepackage{enumerate}
\usepackage{amsfonts}

\usepackage{amscd}
\usepackage[latin1]{inputenc}
\usepackage{enumitem}
\usepackage[mathscr]{eucal}
\usepackage{indentfirst}
\usepackage{graphicx}
\usepackage{graphics}
\usepackage{pict2e}
\usepackage{epic}
\usepackage{esint}
\usepackage[margin=2.9cm]{geometry}
\usepackage{epstopdf} 
\usepackage{verbatim}
\usepackage{cancel}
\usepackage{amssymb}
\usepackage{xcolor}
\usepackage{dsfont}
\usepackage{hyperref}
\usepackage{physics}
\usepackage{amsmath}

\usepackage{tikz}
\usepackage{float}
\usepackage{pgfplots}

\allowdisplaybreaks

\definecolor{green}{rgb}{0.0, 0.5, 0.5}
\definecolor{yellow}{rgb}{0.5, 0.5, 0}
\definecolor{lgray}{gray}{0.9}
\definecolor{llgray}{gray}{0.95}
\definecolor{lllgray}{gray}{0.975}

\usepackage{hyperref}

\setitemize{itemsep=5pt, topsep=7pt}
\numberwithin{equation}{section}

\newtheorem{definition}{Definition}[section]
\newtheorem{lemma}[definition]{Lemma}

\newtheorem{proposition}[definition]{Proposition}

\newtheorem{theorem}[definition]{Theorem}
\newtheorem{corollary}[definition]{Corollary}

\theoremstyle{remark}%
\newtheorem{example}{Example}%
\newtheorem{remark}{Remark}%

\numberwithin{remark}{section}

\def\cN{\mathcal{N}}
\def\a0{{\rm a}_0}

\def\Im{\mathrm{Im}}

\newcommand{\vertiii}[1]{{\left\vert\kern-0.25ex\left\vert\kern-0.25ex\left\vert #1 
		\right\vert\kern-0.25ex\right\vert\kern-0.25ex\right\vert}}

\newcommand{\inn}[2]{\left\langle#1,\,#2\right\rangle}

\DeclareMathOperator{\supp}{supp}

\newcommand{\const}{\text{const.}}
\newcommand{\Lap}{\Delta}
\newcommand{\di}{\partial}

\newcommand{\br}[1]{\bigl\langle#1\bigr\rangle}

\newcommand{\eps}{\epsilon}

\newcommand{\al}{\alpha}

\newcommand{\Ac}{\mathcal{A}}
\newcommand{\Zb}{\mathbb{Z}}
\newcommand{\Rb}{\mathbb{R}}
\newcommand{\N}{\mathcal{N}}
\newcommand{\cL}{\mathcal{L}}

\newcommand{\E}{\mathcal{E}}

\newcommand{\W}{\mathcal{W}}

\newcommand{\Om}{\Omega}

\renewcommand{\l}{\lambda} 

\newcommand{\sbr}[1]{\left[#1\right]}
\newcommand{\Set}[1]{\left\{#1\right\}}

\newcommand{\od}[2]{\ensuremath{\frac{d#1}{d#2}}}
\newcommand{\md}[6]{\ensuremath{
		\ifinner
		\tfrac{\partial{^{#2}}#1}{\partial{#3^{#4}}\partial{#5^{#6}}}
		\else
		\tfrac{\partial{^{#2}}#1}{\partial{#3^{#4}}\partial{#5^{#6}}}
		\fi
}}
\newcommand{\del}[1]{\left(#1\right)}
\newcommand{\thmref}[1]{Theorem~\ref{#1}}

\newcommand{\secref}[1]{Section~\ref{#1}}
\newcommand{\lemref}[1]{Lemma~\ref{#1}}
\newcommand{\propref}[1]{Proposition~\ref{#1}}

\newcommand{\figref}[1]{Figure~\ref{#1}}
\newcommand{\corref}[1]{Corollary~\ref{#1}}

\DeclareMathOperator{\dG}{\mathrm{d}\Gamma}

\DeclareMathOperator{\Rem}{Rem}

\newcommand{\xidel}[1]{\inn{\xi}{{#1}\xi}}

\renewcommand{\cp}{\mathrm{c}}
\newcommand{\1}{\mathds{1}}
\newcommand{\Fl}{\mathcal{F}_{\perp \varphi_t}^{\leq N}}
\newcommand{\cE}{\mathcal{E}}

\newcommand{\ran}{\rangle}
\newcommand{\lan}{\langle}

\title[Propagation bounds on Bose-Einstein condensates]{Local enhancement of the mean-field approximation for bosons}

\author{Marius Lemm}
\address{Marius Lemm, Department of Mathematics, University of T\"ubingen, 72076 T\"ubingen, Germany }
\email{marius.lemm@uni-tuebingen,de}

\author{Simone Rademacher}
\address{Simone Rademacher, Department of Mathematics, LMU Munich,  Theresienstrasse 39, 80333 Munich, Germany }
\email{simone.rademacher@math.lmu.de}

\author{Jingxuan Zhang}
\address{Yau Mathematical Sciences Center\\
	Tsinghua University\\
	Haidian District\\
	Beijing 100084, China }
\email{jingxuan@tsinghua.edu.cn}
\begin{document}
	
	\begin{abstract}
		We study the quantum many-body dynamics of a Bose-Einstein condensate (BEC) on the lattice in the mean-field regime. We derive a local enhancement of the mean-field approximation: At positive distance $\rho>0$ from the initial BEC, the mean-field approximation error at time $t\leq \rho/v$ is bounded as $\rho^{-n}$, for arbitrarily large $n\geq 1$.    
		This is a consequence of new ballistic propagation bounds on the fluctuations around the condensate.  To prove this, we develop a variant of the ASTLO (adiabatic spacetime localization observable) method for the particle non-conserving generator of the fluctuation dynamics around Hartree states. 
		
		

	\end{abstract}
	
	\maketitle 
	\date{\today}
	
	\section{Introduction}
	We study the quantum many-body dynamics of a system of bosons on the lattice. 
	We formulate the problem as usual in terms of the algebraic operator formalism of creation and annihilation operators, i.e., we consider the symmetric Fock space 
	\begin{align}
		\label{def:Fockspace}
		\mathfrak{F}(\ell^2(\mathbb Z^d))
		=\mathbb C\oplus \bigoplus_{N\geq 1} (\ell^2(\mathbb Z^d))^{\otimes_s N}
	\end{align}
	with $d\geq 3$ representing the spatial dimenion and $\otimes_s N$ denotes the $N$-fold symmetritzed tensor product. On the Fock space, we consider solutions of the $N$-body Schr\"odinger equation
	\begin{align}
		\label{eq:Schroe}
		i \partial_t \psi_{N,t} = H_N \; \psi_{N,t}
	\end{align}
	with self-adjoint Hamiltonian operator given by 
	\begin{align}\label{HNdef}
		H_N =  \sum_{\substack{x,y \in \Zb^d  \\ x \sim y}  } (-\Delta_{x,y}) \;  a_x^*a_y+ {\frac{\l}{ 2N}}\sum_{\substack{x \in \Zb^d  }} \; a_x^*a_x^*a_xa_x \; . 
	\end{align}
	The main large parameter is $N$, the particle number. The interaction parameter $\lambda>0$ is a (small) order-$1$ constant. The prefactor $\frac{1}{N}$ in front of the second sum corresponds to the paradigmatic \textit{mean-field scaling}. 
	
	Since one is interested in the regime of large particle number $N$, eq.~\eqref{eq:Schroe} is effectively a \textit{linear PDE in extremely large dimension}. Understanding its behavior is an instance of the \textit{quantum many-body problem} and it requires the development of mathematical techniques that are quite different from standard PDE theory, because the dimension is so large. 
	

	Mean-field dispersive PDEs of the type \eqref{eq:Schroe} with Hamiltonian operator \eqref{HNdef} constitute a central paradigm in the mathematical study of quantum many-body systems, because they strike a good compromise between being genuinely interacting and not explicitly solvable, but still amenable  to rigorous mathematical analysis. Essentially, the prefactor $\frac{1}{N}$ in \eqref{HNdef} makes the interaction very weak and so the system can be ``expanded'' in a certain sense. As a consequence, bosons in the mean-field regime and their more singular siblings, especially bosons in the Gross-Pitaevskii regime, have been widely studied in the mathematical physics community.
	After breakthroughs by Erd\H{o}s-Schlein-Yau \cite{erdHos2009rigorous,erdos2010derivation} and Rodnianski-Schlein \cite{RS} on the derivation of an effective nonlinear Hartree equation describing the BEC, a lot of effort in the past 15 years has been devoted to many refinements of the dynamical description (e.g., the norm approximation by precisely describing the dynamical fluctuations through Bogoliubov theory); see \cite{AKS,BPS,benedikter2015quantitative,brennecke2019gross,bossmann2022beyond,chen2016klainerman,chen2019derivation,deuchert2023dynamics,DL,GMM,GMMa,knowles2010mean,Kuz,RL23,lewin2015fluctuations,mitrouskas2019bogoliubov,nam2017bogoliubov,pickl2011simple,RS22} and the review \cite{napiorkowski2023dynamics}.

	
	\subsection{First result: Locally enhanced mean-field approximation}
	The aforementioned results imply that the dynamical evolution of an initial BEC $\psi_{N,0}=\varphi_0^{\otimes N}$ can be described for large $N$-values in terms of an $N$-independent nonlinear evolution equations in $d$ dimensions.
	Consider an initially factorized state
	\begin{align}
		\label{purFac}
		\psi_{N,0} = \varphi_0^{\otimes N}. 
	\end{align}
	Due to the interaction term, the factorization is not preserved along the time evolution given by the Schr\"odinger equation \eqref{eq:Schroe}. Nonetheless, a macroscopic fraction of the particles still occupies the same one-particle state, called the condensate, whose dynamics is described by the nonlinear Schr\"odinger equation
	\begin{align}
		\label{Hartree}
		i \partial_t \varphi_t = h_{\varphi_t}\varphi_t \; \quad \text{with} \quad h_{\varphi_t}: =   - \Delta  +  {\l\vert \varphi_t \vert^2}  .
	\end{align}
	Denoting the reduced one-particle density matrix of the time-evolved quantum state $\psi_{N,t}$, by $\gamma_{\psi_{N,t}}$ and considering a one-particle observable with operator kernel $O=O(x,y)$, one has the mean-field approximation (see e.g., \cite{DL,RS})
	\[
	\big\vert \Tr \big(\big(\tfrac{1}{N}\gamma_{\psi_{N,t}} -  \vert \varphi_t \rangle \langle \varphi_t \vert\big)  O\big) \big\vert \leq \frac{C\|O \|_{\rm op} }{N} ,   \qquad t\geq 0.
	\]
	
	These prior approximation results probe the quantum gas \textit{globally} in space. However, we may expect that the approximation can be \textit{locally enhanced} based on the physical principle of \textit{spatial locality} which suggests that quantum information (and more general physical quantities) should only propagate at most with some bounded speed,  up to small errors.
	
	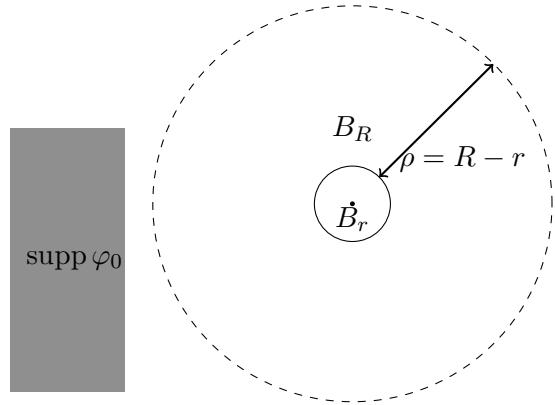
\begin{figure}[H]
		\centering
		\begin{tikzpicture}[scale=.5]
			\draw (0,0) circle (1) node at (0,-.4) {$B_r$};
			
			\filldraw (0,0) circle (0.05);
			
			\draw[dashed] (0,0) circle (3.5*1.5);
			


			\draw[thick,<->] (0.7071,0.7071) to (2.47*1.5,2.47*1.5);
			
			\node at (0,2) {$B_R$};
			
			\node at (2.9,1.2) {$\rho=R-r$};
			
			
			
			\fill[gray!88] (-6,-5) rectangle (-9,2);
			
			\node at (-7.3,- 1.5) {$\supp\varphi_0$ };
		\end{tikzpicture}
		\caption{The geometry in the first result (see \thmref{thm:trace}). If the initial condensate is supported away from $B_R$ and the test observable $O$ is supported in $B_r$, then  the mean-field approximation error $\mathrm{Tr}((\gamma_{\psi_{N,t}}-N\ket{\varphi_t}\bra{\varphi_t})O)$ is bounded by $\rho^{-n}$ for times $t\leq \rho/v$ where $\rho=R-r$ is the distance that has to be traversed.}\label{figBall1}
		
	\end{figure}
	
	By turning this principle into mathematical bounds, we improve the mean-field approximation at positive distances $\rho>0$ away from the initial BEC for times $t$ that are short compared to $\rho$. More precisely, in \thmref{thm:trace} below, we establish that for initial states $\psi_{N,0}= \varphi_0^{\otimes N}$ which are supported away from the origin, i.e.,
	\[
	\varphi_0 (x) = 0, \quad \text{for all} \quad x \in B_{R}\; ,
	\]
	the mean-field approximation at smaller distances distance $r< R$ from the origin, is significantly enhanced if the time is short compared to $\rho=R-r$, the distance that has to be traversed.

	Consider the situation shown in Figure \ref{figBall1} and let $O$ be any local observable acting on the smaller ball $B_r$, i.e.,
	\[
	O(x,y) = \mathds{1}_{x \in B_r} O(x,y) \mathds{1}_{y \in B_r}.
	\]
	Our first main result says that the mean-field approximation error satisfies the improved bound
	\begin{equation}\label{eq:main1}
		\big\vert \Tr \big(\big(\tfrac{1}{N}\gamma_{\psi_{N,t}} -  \vert \varphi_t \rangle \langle \varphi_t \vert\big)  O\big) \big\vert \leq
		\frac{C \|O \|_{\rm op} }{N}
		\frac{ 1  }{\rho^n},   \qquad \textnormal{for } t \leq  \frac{\rho}{ v},
	\end{equation}
	for \textit{any} $n\geq 1$. Here, $v$ is a state-independent, order-$1$ constant that plays the role of a velocity bound. In other words, we prove that for times $t \leq  \frac{\rho}{ v}$, the mean-field approximation error at distance $\rho>0$ from the initial condensate is enhanced to $\mathcal O(\rho^{-\infty})$.

	\subsection{Second result: Propagation bound on fluctuations}
	The concrete way in which we exploit the general principle of spatial locality is by showing that fluctuations  around the Bose-Einstein condensate effectively propagate at bounded speed $v>0$. That is, we derive a ``ballistic'' propagation bound on the fluctuations (a propagation bound is called ``ballistic'', if it is effective for times proportional to distance). This ballistic propagation bound then implies Theorem \ref{thm:trace} because it says that the fluctuations did not have sufficient time to propagate from $B_R^c$ to $B_r$, for times that are short compared to $(R-r)/v$,

	The propagation bound concerns the number of excitations, as encoded in a self-adjoint operator $\mathcal N^+$. The quality of the bound on the number of fluctuations $\mathcal N^+$ directly translates to the quality of the mean-field approximation.
	
	The standard bounds on $\mathcal N^+$ are of global nature. For instance, the influential works \cite{RS,lewin2015fluctuations} proved 
	\[
	\frac{\langle \psi_{N,t},(\mathcal N^++1)\psi_{N,t}\rangle}{\langle \psi_{N,0},(\mathcal N^+ +1)\psi_{N,0}\rangle}\leq Ce^{Ct} 
	\]
	where $\psi_{N,t}$ solves \eqref{eq:Schroe}. Most importantly, the bound says that the number of excitations is bounded independently of $N$. The term $e^{Ct}$ is somewhat undesirable because it rapidly grows in time. At the same time, the exponential growth appears to be very natural because it arises from a Gr\"onwall argument and indeed it is probably optimal in general, i.e., without any further assumptions on the nonlinear effective evolution equation.
	
	One physically natural assumption is that of dispersion (validity of nonlinear Strichartz estimates). This allows to completely remove the growth in time, as shown recently by Dietze-Lee \cite{DL} (see also \cite{Lee})
	\[
	\frac{\langle \psi_{N,t},(\mathcal N^++1)\psi_{N,t}\rangle}{\langle \psi_{N,0},(\mathcal N^+ +1)^2\psi_{N,0}\rangle}\leq C \; .
	\]
	
	Our approach is to consider the \textit{local} fluctuation numbers $\mathcal N_X^+$ with $X\subset\mathbb Z^d$. We show  that in the situation of Figure \ref{figBall1}, the fluctuations inside $B_r$ at time $t$ are $\mathcal O(\rho^{-\infty}) $ for $t\leq \frac{\rho}{v}$. More precisely, we derive the enhanced local bound
	 	\begin{equation}\label{eq:main2}
			\frac{\langle \psi_{N,t},(\mathcal N^+_{B_r}+1)\psi_{N,t}\rangle}{\langle \psi_{N,0},(\mathcal N^+  +1)^2\psi_{N,0}\rangle}\leq \frac{C}{\rho^n},
	\end{equation} 
	for $t\leq \frac{\rho}{v}$ and arbitrarily large $n\geq 1$. The precise statement is in Theorem \ref{thm:nloc} below. We point out that Theorem \ref{thm:nloc} is actually more general than the situation depicted in Figure \ref{figBall1}, because it does not require the initial state to be a pure condensate.
	
	\subsection{Related literature on many-body propagation bounds}
	Our idea is to exploit the general principle of spatial locality, which suggests that physical quantities effectively exhibit a maximal speed. In relativistic systems, this maximal speed is the speed of light, but this is neither mathematically valid nor physically meaningful (simply because the speed of light is so large)  in a non-relativistic model as the one described by \eqref{eq:Schroe} and \eqref{HNdef}. A suitable analog of the speed of light for many-body lattice systems goes by the name of ``Lieb-Robinson bounds'' \cite{lieb1972finite}, which showed how spatial locality can be used in a many-body context. Note in passing that the standard Lieb-Robinson bounds are concerned with  quantum spin systems, which have bounded interactions and therefore do not apply in our setting.

	Several recent works have developed propagation bounds (including Lieb-Robinson bounds) for Bose-Hubbard  Hamiltonians    \cite{FLS,FLSa,kuwahara2021lieb,kuwahara2024effective,kuwahara2024enhanced,LRSZ,LRZ,schuch2011information,yin2022finite}.
	These bounds actually apply to our Hamiltonian (which is a Bose-Hubbard Hamiltonian with the interaction scaled by $N^{-1}$). However, they are not able (nor designed) to give results like ours here. The reason is that 
	these other approaches are designed to treat all the bosons alike. This has the consequence that they are not able to separate the Bose-Einstein condensate, which is comprised of $\sim N$ many bosons, from the fluctuations, which are comprised of $\mathcal O(1)$ many bosons. Therefore, these bounds would lead to an additional $N$-factor on the right-hand sides of \eqref{eq:main1} and \eqref{eq:main2}, which would no longer be a useful bound. Therefore, the bounds developed in \cite{FLS,FLSa,kuwahara2021lieb,kuwahara2024effective,kuwahara2024enhanced,LRSZ,LRZ,schuch2011information,yin2022finite} are inconclusive for what we undertake here. \textit{This is why we need to develop new propagation bounds that are specially tailored to tracking the fluctuations.} Propagation bounds for the nonlinear Hartree equation were proved in \cite{arbunich2023maximal,huang2021uncertainty} and here we treat for the first time the fluctuations.

	As far as we are aware, propagation bounds on fluctuations have not played a large role so far in the dynamical description of bosons in the mean-field regime. A notable exception is a work by Erd{\"o}s-Schlein \cite{erdHos2009quantum} which used Lieb-Robinson type ideas between different particle number sectors inspired by \cite{nachtergaele2007lieb}, but without considering the spatial structure.  
	
	Methodologically, our result here builds on the  developments on bosonic propagation bounds, but with a new twist. In a nutshell, our key contribution is to demonstrate how the ASTLO (adiabatic spacetime localization observables) method that was originally developed for Bose-Hubbard Hamiltonians (see \cite{FLS,FLSa,LRSZ,LRZ}) can be adapted to the \textit{particle non-conserving} generator of the fluctuation dynamics. 
	
		Lastly, 	we would also like to mention that  semiclassical propagation bounds have played an important role in the mean-field description of fermions; see, e.g., \cite{elgart2004nonlinear,benedikter2014mean,lafleche2023strong,fresta2023effective}. These bounds are again concerned with exploiting the general principle of spatial locality, but in a physical and mathematical setting that is quite different from ours.
	
	
	\subsection{Setup and assumptions}
	
	We consider the Hamiltonian $H_N$ from \eqref{HNdef}.
	Here $\abs{\l}\le \l_0$ for some $\l_0>0$ to be determined later,
	and $-\Lap_{x,y}$ stands for the discrete Laplacian 
	\begin{align}
		\label{def:Laplace}
		-\Delta_{x,y} = \begin{cases}
			-1 & \text{if} \quad x\sim y,\\
			2d & \text{if}\quad  x=y,\\
			0& \text{otherwise}.
		\end{cases}
	\end{align}
	Here, $x\sim y$ denotes nearest neighbors on $\mathbb Z^d$. 
	
	

	
	For $t\geq 0$, we denote by $\varphi_t$  the solution to the Hartree equation \eqref{Hartree} with initial data $\varphi_0$,    $\psi_{N,t}$ the solution to the Schr\"odinger equation \eqref{eq:Schroe} with initial data $\psi_{N,0}$, and $\gamma_{\psi_{N,t}}$ the associated reduced one-particle  density matrix satisfying 
	$$\Tr  \gamma_{\psi_{N,t}}   O 
	=  \langle \psi_{N,t}, \dG (O) \psi_{N,t} \rangle,$$
	for any $1$-particle operator $O$, where $\dG$ denotes the second quantization map.
	  We abbreviate ${\ell^p}\equiv{\ell^p(\Zb^d)}$ for $1\le p\le \infty$,  $B_r:=\Set{x\in\Zb^d:\abs{x}\le r}$ for $r>0$, and  $\norm{\cdot}_{\mathrm{op}}$  for the operator norm  on $\ell^2$.

	We let  $d\ge3$ and assume
	the validity of the following Strichartz-type \textit{dispersive estimate (DE)} for the solution $\varphi_t$ of the Hartree equation \eqref{Hartree}:
	\begin{enumerate}[label=$\mathrm{(DE)}$]
		\item \label{disCond}The initial state $\varphi_0\in\ell^1$ with $\norm{\varphi_0}_{\ell^2}=1$, 
		and 	there exist $c,\,T>0$ independent of $\varphi$ such that
		\begin{align}\label{dis}
			\int_0^t \norm{\varphi_s}_{\ell^\infty}\,ds\le c \norm{\varphi_0}_{\ell^1},\qquad t\leq T.
		\end{align}
	\end{enumerate}
	
	The dispersive estimate  \ref{disCond}, that we assume in the following, can be verified for small $\lambda$.	The proofs are in \secref{sec43} and  Appendix \ref{App1} and some examples are given below.
	\begin{example}\label{ex1}
		Let $d\ge4$. Then there exist constants $\l_0,\,c>0$ depending only on $d$ s.th.~if $\abs{\l}\le\l_0$, then \eqref{dis} holds for all $t\ge0$, i.e., $T=\infty$.
	\end{example}
	
	\begin{example}\label{ex2}
		Let $d=3$. Then there exist absolute constants $\l_0,\,c>0$ s.th.~if $\abs{\l}\le\l_0$ and 	$T:= { {e^{1/\sqrt{\abs{\l}}}-1}}$, then \eqref{dis} holds.
	\end{example}

	To summarize, in dimensions $d\ge4$, Assumption \ref{disCond} holds globally in time and in $d=3$ it holds up to a large time $T=e^{1/\sqrt{\abs{\l}}}-1$, which is not very restrictive for us, since we are anyway only interested in bounds that hold for sufficiently short times. For additional information about Assumption \ref{disCond}, see Remark \ref{remark:disCond} at the end of this section. 

	\subsection{Results}
	Our first theorem provides a local enhancement of the mean-field approximation in the geometric situation schematically depicted in \figref{figBall1}.

	\begin{theorem}[Locally enhanced mean-field approximation] \label{thm:trace} 
		Let $d\ge3$ and suppose that condition \ref{disCond} hold.  Assume  
		\begin{itemize}
			\item  $\psi_{N,0}= \varphi_0^{\otimes N}$ is purely factorized;
			\item $\varphi_0$ satisfies, for some $r>0,\,\rho >4d$,
			\begin{align}\label{phi0cond}
				\varphi_0 (x) = 0, \quad \text{for all} \quad x \in B_{r+\rho}\; . 
			\end{align}
		\end{itemize}

		Then for any $4d<v\le\rho$ and any integer $n\geq 1$, there exists a constant $C>0$ depending only on $n,\,v,\,d$, $\abs{\l}$, the constant $c$ in \eqref{dis}, and the $\ell^1$ norm of $\varphi_0$, such that for any bounded local operator $O$ acting on $\ell^2$ with kernel satisfying 
		\begin{align}
			\label{ass:O}
			O(x,y) = \mathds{1}_{x \in B_r} O(x,y) \mathds{1}_{y \in B_r}  ,
		\end{align}
there holds
		\begin{align}
			\big\vert \Tr \del{\big(\gamma_{\psi_{N,t}} - N \vert \varphi_t \rangle \langle \varphi_t \vert\big)  O} \big\vert \leq\frac{ C \|O \|_{\rm op}    }{\rho^n},   \quad t \leq  \frac{\rho}{ v}  \; .\label{110}
		\end{align}
	\end{theorem}
	
	This theorem is proved in \secref{sec:trace}.	Comparing with existing result (e.g.~\cite{RS, Lee, DL}), it says that mean-field approximation (l.h.s.~of \eqref{110}) is small for fixed $N$, as long as    {the system is probed sufficiently far away from the initial states $\varphi_0$.}
	\begin{remark}
		The lower bound on the velocity $v>4d$  in Theorem \ref{thm:trace} arises as follows: the threshold value $2d$ naturally arises as an upper bound on the $1$-body momentum operator $i[-\Lap,\abs{x}]$ (this fact can be easily verified by Schur's test) and two successive propagation bounds are used in the proof. By contrast, the upper bound	$v\le\rho$ is purely technical and can be replaced by $v \le \mu\rho$ for \thmref{thm:trace}  for any $\mu>0$, with the final estimate \eqref{110}   depending on $\mu$. The situation is similar for \thmref{thm:nloc}. 
	\end{remark}


	\medskip

	We come to our second main result, the propagation bound on the number of fluctuations orthogonal to the condensate,
	\begin{align}\label{N+def}
		\mathcal{N}^+ (t) = \sum_{i=1}^N q_t^{(i)},
	\end{align}
	where $q_t^{(i)}$ is an operator on  $\ell_s^2( \Zb^{dN}):=[\ell^2(\Zb^d)]^{\otimes _s N}$, acting as $q_t = 1 - \vert \varphi_t \rangle \langle \varphi_t \vert$ on the $i$-th particle and as identity on all other particles. We prove that the \textit{local} number of excitations 
	\begin{align}
		\mathcal{N}_{B_r}^+ (t) = \sum_{i=1}^N  \big( q_t \mathds{1}_{B_r} q_t \big)^{(i)}\label{Nr} 
	\end{align}
	satisfies a propagation bound that depends on the initial  geometric configuration of $\psi_{N,0}$ and $\varphi_0$; see \figref{figBall}.


	{	\begin{theorem}[Propagation bound on fluctuations]
			\label{thm:nloc} 
			Let $d\ge3$ and let condition \ref{disCond} hold.  Assume  
			that the initial state is of the more general form
				\[\psi_{N,0} = \sum_{j=0}^N \varphi_0^{\otimes (N-j)} \otimes_s \xi_{0}^{(j)},
				\]
			where $\xi_{0}^{(j)} \in  \big( \ell^2_{\perp \varphi_0}( \mathbb{Z}^d)\big)^{\otimes_s N}$ and $\ell^2_{\perp \varphi_0}( \mathbb{Z}^d)$ denotes the orthogonal complement of $\varphi_0$ in $\ell^2( \mathbb{Z}^d)$,
			and 
			\begin{itemize}
				
				\item $\varphi_0$ satisfies, for some $r>0,\,\rho>2d$, 
				\begin{align}\label{phi0cond'}
					\varphi_0 (x) = 0, \quad \text{for all} \quad x \in B_{r+2\rho}\; ; 
				\end{align}
				\item  $\psi_{N,0}$ has no fluctuations in $B_{r+\rho}$, i.e. 
				%
				\begin{align}
					\label{psi0Cond}
					\big\langle \psi_{N,0}, \; \mathcal{N}_{B_{r+\rho}}^+ (0) \;  \psi_{N,0} 
						=0.
				\end{align}
			\end{itemize}
			Then for any $2d<v\le \rho$ and any integer $n\geq 1$, there exists a constant $C>0$ depending only on $n,\,v,\,d$, $\abs{\l}$, the constant $c$ in \eqref{dis}, and the $\ell^1$ norm of $\varphi_0$, such that
			\begin{align}
				\label{115}
				\big\langle \psi_{N,t}, \; \mathcal{N}_{B_r}^+ (t) \;  \psi_{N,t} \big\rangle \le \frac{C }{\rho^n}\br{\del{\cN^+(0)+1}^2}_0,\quad t\le\frac{\rho}{v}.
			\end{align}
	\end{theorem}}
	This theorem is proved in \secref{sec5}.
	
		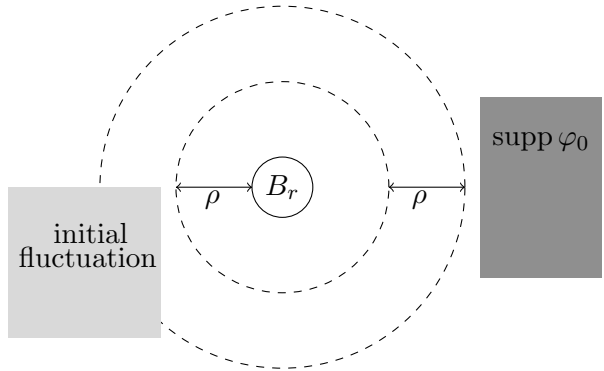
\begin{figure}[h!]
		\centering
		\begin{tikzpicture}[scale=.4]
			\draw (0,0) circle (1) node at (0,0) {$B_r$};
			
			
			\draw[dashed] (0,0) circle (3.5);
			
			\draw[dashed] (0,0) circle (6);
			
			\draw[<->] (-1,0) -- (-3.5,0)  ;
			
			\draw[<->] (3.5,0) -- (6,0)  ;


			
			
			\node at (-2.3,-0.4) {$\rho$};
			\node at (4.5,-0.4) {$\rho$};
			
			
			\fill[gray!88] (6.5,-3) rectangle (10.5,3);
			
			\fill[gray!30] (-4,-5) rectangle (-9,0);
			
			\node at (8.5, 1.5) {$\supp\varphi_0$};
			\node at (-6.3,- 1.5) { initial };
			\node at (-6.4,- 2.3) { fluctuation};
		\end{tikzpicture}
		\caption{The assumption on the geometry in \thmref{thm:nloc}, which is a weaker assumption than that for Theorem \ref{thm:trace}.  Estimate \eqref{115} says that if the initial $\varphi_0$ is supported $2\rho$-away from $B_r$ and there are initially no fluctuations on $B_{r+\rho}$, then  the local fluctuation number in $B_r$ is suppressed up to time $t\le\rho/v$.}\label{figBall}  
		
	\end{figure}

	Physically, Theorem \ref{thm:nloc} controls two processes in which fluctuations can move: either (a) by the condensate moving through space and fluctuations being created from the condensate or (b) by existing fluctuations moving through space. Of course, these two processes are not independent from each other and controlling their joint propagation as described by the generator of the fluctuation dynamics is at the heart of the proof. 


	Notice that Theorem \ref{thm:nloc} does not require $\psi_{N,0}$ to be globally factorized as in \eqref{purFac}. It only requires that the initial condensate part and the initial fluctuations are supported away from the ball $B_r$ where we probe the system as shown in \figref{figBall}. 
	If one assumes the stronger condition that $\psi_{N,0}$ is purely factorized, i.e., \eqref{purFac}, then condition \eqref{psi0Cond} holds for all $\rho>0$ and \eqref{115} improves to
	$$				\big\langle \psi_{N,t}, \; \mathcal{N}_{B_r}^+ (t) \;  \psi_{N,t} \big\rangle \le \frac{C  }{\rho^n},\qquad t\le \frac\rho v.$$

	
	We close the presentation of the main results with some remarks on the dispersive estimate  \ref{disCond} that we assumed throughout.

	\begin{remark}[Remark on the dispersive estimate  \ref{disCond}]\label{remark:disCond}
		\;
		\begin{enumerate}[label=(\roman*)]
			
			\item 
			The reason that for $d=3$ the dispersive estimate  \ref{disCond} does not necessarily hold globally in time is not technical, but rather has to do with the discrete Laplacian truly having  slightly weaker dispersive properties than its continuum counterpart. Indeed, the natural long-time scaling of the free evolution $\|e^{it\Delta}f\|_\infty$ is only $t^{-d/3}$ on the lattice compared to the $t^{-d/2}$-decay that is familiar from the continuum; see \cite{SK,HY} for explanations how $t^{-d/3}$ arises from properties of the Fourier symbol. For us, the presence of $T= {e^{1/\sqrt{\abs{\l}}}-1}$ in $d=3$ does not play an important role, because we are anyway only interested in proving a result for reasonably short times.
			
			\item The condition $\abs{\l}\le \l_0$ is also used to ensure global-wellposedness. Indeed, through the rescaling $\varphi_t \mapsto \sqrt{\abs{\l}}\varphi_t$ of the condensate wave function, the condition on $\l$  can be replaced by smallness condition in $\ell^2$-norm of the initial data. 
			
			\item 
			If we assume well-posedness of the Hartree equation, then it is also possible to drop \ref{disCond} altogether. In this case, the bounds hold for all $d\geq 1$ with  the constant $C$ in \eqref{110},  \eqref{115} replaced by $Ce^{Ct}$; see \eqref{512}, \eqref{110'}. 
			This bound  still locally enhances prior results for times $t\ll \log (R-r)$.
			

			\item 
			The $\ell^1$-norm of initial data enters the main results \eqref{110} and \eqref{115} through the  assumption of dispersive estimate \ref{disCond}. We typically consider the initial condensate function $\varphi_0$, which is always $\ell^2$-normalized, to be spread over an order-$1$ length scale and so $\|\varphi_0\|_1$ plays the role of an order-$1$ constant. 

		\end{enumerate}

	\end{remark}
	\subsection{Organization of the paper}
	In \textit{\secref{sec:flucdyn}}, we present preliminary results that are used repeatedly in the sequel, including the definition for the fluctuation dynamics, its basic properties,  and some commutator bounds. 
	In \textit{\secref{secLocFlucEst}}, we prove a key local fluctuation estimate that consistutes the main ingredient in the proof of the main results.	
	In \textit{\secref{sec4}}, combining certain commutator bounds for the fluctuation dynamic and the dispersive estimate  for the condensate dynamic \eqref{Hartree}, we prove a priori $\mathcal O(1)$ estimates on the total fluctuation. 	
	In \textit{Sections \ref{sec5}} and \textit{\ref{sec:trace}}, we use the preceding results to prove the main Theorems \ref{thm:trace} and \ref{thm:nloc}. 
	In the three Appendices, we collect relevant dispersive estimates for the discrete Hartree equations and the proofs of some technical lemmas.

	\section{Fluctuations around the Hartree evolution} 
	\label{sec:flucdyn} 
	
\subsection{Truncated Fock space $\Fl$}
	To control the evolution of quantum fluctuation, we need to factor out contributions of the condensate  dynamics given by the Hartree equation \eqref{Hartree} . To this end we observe that {given a solution $\varphi_t,\,t\ge0$ evolving according to \eqref{Hartree},} any $\psi_N \in \ell^2_s( \Zb^{dN})$ has a unique decomposition as 
	\begin{align}
		\psi_N = \sum_{j=0}^N \varphi_t^{\otimes (N-j)} \otimes_s \xi_{t}^{(j)}, \quad \text{where} \quad \xi_{t}^{(j)} \in  \ell^2_{\perp \varphi_t}( \mathbb{Z}^d)^{\otimes_s j}
	\end{align}
	into products of the condensate wave function $\varphi_t$ and excitations $\xi_t^{(j)}$ that are orthogonal to products of the condensates. There is a unitary, first introduced in \cite{LNSS15},  
	\begin{align}
		\mathcal{U}_{N,t} :\ell^2_s( \Zb^{dN}) \rightarrow \bigoplus_{j=0}^N \ell^2_{\perp \varphi_t}( \mathbb{Z}^d)^{\otimes_s j} := \mathcal{F}_{\perp \varphi_t}^{\leq N}, \quad \mathcal{U}_{N,t}  \psi_N =\lbrace \xi_t^{(0)}, \dots, \xi_t^{(N)} \rbrace =:  \xi_{N,t} \label{UNtDef}
	\end{align}
	mapping an $N$-particle wave function $\psi_N$ to  excitations $\xi_{N,t}$. Note that in contrast to the standard bosonic Fock space \eqref{def:Fockspace} that is built over $\ell^2( \mathbb{Z}^d)$, the truncated Fock space $\Fl$ is built over $\ell^2_{\perp \varphi_t}( \mathbb{Z}^d)$, the orthogonal complement of $\varphi_t$ in $\ell^2( \mathbb{Z}^d)$. Moreover, $\Fl$  contains sequences $\xi_{N,t}$ of at most $N$ elements. 
	
		On the truncated Fock space $\mathcal{F}_{\perp \varphi_t}^{\leq N}$ we furthermore have for any operator $A$ on $\ell^2( \mathbb{Z}^d)$ that 
	\begin{align}
		\label{eq:id2}
		\dG (q_t A q_t) = \dG (A).
	\end{align}
	Specifically, for the local excitation number operator (c.f.~\eqref{Nr})
	\begin{align}
		\mathcal{N}_{X}^+ (t):= \sum_{i=1}^N  \big( q_t \mathds{1}_{X} q_t \big)^{(i)} ,\quad X\subset \Zb^d,
	\end{align}
  the following identity holds on $\Fl$ for any $X$: 
	\begin{align}
		\mathcal{N}_X^+ (t) = \mathcal{N}_X := \sum_{z \in X} n_z  \; , \quad \text{with} \quad n_z  := a_z^*a_z \label{Nid} \; . 
	\end{align}
	
			\subsection{Fluctuation dynamic $\W_N(t;s)$} 
	We are interested in fluctuations around (products of) the condensate that we study through the so-called fluctuation dynamics on $\Fl$:
	\begin{align}
		\mathcal{W}_N (t;s) = \mathcal{U}_{N,t} e^{-iH_N (t-s)}\mathcal{U}_{N,s}^* \label{def:W},
	\end{align}
	which satisfies 
	\begin{align}
		i \partial_t \mathcal{W}_N (t;s) = \mathcal{L}_N (t) \mathcal{W}_N(t;s) ,
	\end{align}
	with generator 
	\begin{align}
		\label{def:L}
		\mathcal{L}_N (t) =  \mathcal{U}_{N,t}^* H_N \mathcal{U}_{N,t} + \big( i \partial_t \mathcal{U}_{N,t}^* \big)\mathcal{U}_{N,t}  \; . 
	\end{align}
	
		Recall the global number of excitation is described by the operator $\cN^+(t)$ from \eqref{N+def}.
	We can explicitly compute the generator $\mathcal{L}_{N} (t)$ using that the unitary $\mathcal{U}_N (t)$ acts  on products of creation and annihilation operators for $f,g \in \ell^2_{\perp \varphi_t}( \mathbb{Z}^d)$ as 
	\begin{align}\label{def:b}
		\left\{\begin{aligned}
			\mathcal{U}_{N,t}^* a^*(\varphi_t) a(f) \mathcal{U}_{N,t} =& \sqrt{N - \mathcal{N}^+ (t)} a(f)  =: \sqrt{N} b(f),  \\
			\mathcal{U}_{N,t}^* a^*(f) a(\varphi_t) \mathcal{U}_{N,t} =&  a^*(g)\sqrt{N - \mathcal{N}^+ (t)}  =: \sqrt{N } b^*(g),
		\end{aligned}\right.
	\end{align}
	and
	\begin{align}
		\label{eq:propU}
		\mathcal{U}_{N,t}^* a^*( \varphi_t) a(\varphi_t) \mathcal{U}_{N,t} = N - \mathcal{N}^+ (t), \quad \mathcal{U}_{N,t}^* a^*(f) a(g) \mathcal{U}_{N,t} = a^*(f) a(g) \; . 
	\end{align}	
	In \eqref{def:b} we introduced modified creation and annihilation operators $b^*(f),b(g)$ that, in contrast to $a^*(f),a(g)$, leave the truncated Fock space $\mathcal{F}_{\perp \varphi_t}^{\leq N}$ invariant. This properties comes with the price of modified canonical commutation relations (CCR) 
	\begin{align}
		[b(f), b^*(g)] = \bigg( 1 - \frac{\mathcal{N}^+ (t)}{N} \bigg) \langle g,f \rangle - \frac{1}{N} a^*(g) a(f) , \quad [b^*(f), b^*(g)] = [b(f), b(g)] = 0 ,
	\end{align}
	which, compared to  the standard CCR 
	\begin{align}
		\label{eq:CCR}
		[a(f), a^*(g)] = \langle f, g \rangle, 
	\end{align}
	have a correction term of order $\mathcal O(N^{-1})$.

			\subsection{Basic properties of $\W_N(t)$}

	Based on \eqref{def:b} and \eqref{eq:propU}, explicit computations show that the generator of the fluctuation dynamics \eqref{def:W} is given by 
	\begin{align}
		\label{eq:L-sum}
		\mathcal{L}_N (t) := \mathbb{H} + \sum_{j=1}^N \mathcal{R}_{N,t}^{(j)},
	\end{align}
	where the leading order, quadratic term is given by 
	\begin{align}
		\mathbb{H}  = \dG \big( h_{\varphi_t} + \lambda \widetilde{K}_{1,t} \big) + \frac{\lambda}{2} \sum_{ x \in\Zb^d}  \big[ \widetilde{K}_{2,t} (x) b_x^*b_x^* + \overline{\widetilde{K}}_{2,t} (x) b_x b_x \big] \label{211}.
	\end{align}
	Here, with $J: \ell^2( \Zb^d) \rightarrow \ell^2( \Zb^d)$ denoting the anti-linear operator $Jf = \overline{f}$, we write
	\begin{align}
		\label{tKdef}
		\widetilde{K}_{1,s} =& q_s K_{1,s} q_s,\quad \widetilde{K}_{2,s} = (Jq_sJ) K_{2,s} q_s,\\
		K_{1,s}(x) =& \varphi_s (x) \overline{\varphi}_s (x), \quad K_{2,s} (x) = \varphi_s (x) \varphi_s (x) . \label{tkDef2}
	\end{align}
	The remainder terms are given by 
	\begin{align}
		\label{def:Ri}
		\mathcal{R}_{N,t}^{(1)} =& \frac{\lambda}{2} \dG \big( q_t \big[\vert \varphi_t \vert^2 \varphi_t + \widetilde{K}_{1,t} {-\mu_t}\big] q_t \big)  \frac{1- \mathcal{N}^+ (t)}{N} + \lambda {\frac{\mathcal{N}^+(t) }{\sqrt{N}} b(q_t \vert \varphi_t \vert^2 \varphi_t )} + {\rm h.c.},   \\
		\mathcal{R}_{N,t}^{(2)} =& \frac{\lambda}{\sqrt{N}} \sum_{x \in \Zb^d} \varphi_t (x) a^*(q_{t,x} ) a(q_{t,x} ) b (q_{t,x} ) + {\rm h.c.}, \notag \\
		\mathcal{R}_{N,t}^{(3)} =& \frac{\lambda}{N}\sum_{x \in \Zb^d}  a^*(q_{t,x} ) a^*(q_{t,x} )a(q_{t,x} ) a(q_{t,x} ) \; . \label{213}
	\end{align}
	Here,  in \eqref{def:Ri}, we set \begin{align}
		\label{mutDef}
		2\mu_t := \sum_{x \in \mathbb{Z}^d} \vert \varphi_t (x) \vert^2 \; \vert \varphi_t (y) \vert^2.
	\end{align} 

	In the remainder of this paper, we will derive various bounds on (powers of) the number of excitations $\mathcal{N}^+ (t)$  w.r.t.~to the solution $\psi_{N,t}$ of the Schr\"odinger equation \eqref{eq:Schroe} that is given by 
	\begin{align}
		\langle \psi_{N,t}, \mathcal{N}^+ (t) \psi_{N,t} \rangle = \langle e^{-iH_N t} \psi_{N,0}, \mathcal{N}^+ (t) e^{-iH_N t} \psi_{N,0} \rangle  \; . 
	\end{align}
	With the definition of the fluctuation dynamics $\mathcal{W}_N (t;0)$ in \eqref{def:W} and the observation that 
	\begin{align}
		\label{eq:id1}
		\mathcal{U}_{N,t} \dG (q_t A q_t ) \mathcal{U}_{N,t}^* = \dG (q_t A q_t )
	\end{align}
	for any operator $A$ on $\ell^2( \mathbb{Z}^d)$ from \eqref{eq:propU}, and thus, in particular $\mathcal{U}_{N,t} \mathcal{N}^+ (t) \mathcal{U}_{N,t}^* = \mathcal{N}^+ (t)$, we find  
	\begin{align}
		\langle \psi_{N,t}, \mathcal{N}^+ (t) \psi_{N,t} \rangle = \langle \mathcal{W}_{N}(t;0) \mathcal{U}_{N,0} \psi_{N,0}, \mathcal{N}^+ (t) \mathcal{W}_{N}(t;0) \mathcal{U}_{N,0} \psi_{N,0} \rangle .
	\end{align}

	Lastly, we summarize in the following  lemma various estimates to bound the commutators with the generator of the fluctuation dynamics $\cL_N(t)$:

	\begin{lemma}
		\label{lemma:commutator}  
		For $h: \Zb^d \rightarrow \mathbb{R}$, we have 
		\begin{align}\label{217}
			i \big[ & \mathcal{L}_N (t) - \dG( - \Delta ), \sum_{z \in \Zb^d} h(z) n_z \big] \notag \\
			& \leq 2 \vert \lambda \vert  \sum_{x \in \Zb^d} \vert h(x) \vert \;  \bigg( 2 \vert \varphi_t (x) \vert^2 n_x  + \frac{\vert \varphi_t (x) \vert^3}{\sqrt{N}} \mathcal{N} n_x^{1/2} + \frac{\vert \varphi_t (x) \vert}{\sqrt{N}} n_x (n_x+1)^{1/2} \bigg),
		\end{align}
		as on operator inequality on $\mathcal{F}_{\perp \varphi_t}^{\leq N}$. 
		Moreover, there exists a universal constant $C>0$ such that for $\psi, \xi \in \mathcal{F}_{\perp \varphi_t}^{ \leq N}$, we have
		\begin{align}
			\big\vert \big\langle \xi,  \big[ \mathcal{N}, \mathcal{L}_N (t)- \dG ( -\Delta) \big]  \psi \big\rangle \big\vert &\leq C \vert \lambda \vert  \| \varphi_t \|_{\ell^\infty} \| ( \mathcal{N} + 1) \psi \| \; \| \xi \|  ,\label{eq:comm1}\\ 
			\big\vert \big\langle \xi,  ( \mathcal{N} + 3 )^{-1/2} \big[ \mathcal{N}, \big[ \mathcal{N}, \mathcal{L}_N (t) - \dG ( -\Delta)\big] \big] \psi \big\rangle \big\vert &\leq C \vert \lambda \vert  \| \varphi_t \|_{\ell^\infty} \| ( \mathcal{N} + 1)^{1/2} \psi \| \; \| \xi \| \; .  \label{eq:comm2}
		\end{align}
	\end{lemma}
	This lemma is proved in Appendix  \ref{sec:commutator-estimates}.

	\section{Local fluctuation estimate}\label{secLocFlucEst}
	In this section, we prove local decay estimates for the fluctuation dynamics $\W_N(t;s)$ defined in \eqref{def:W}.  Within this section we will work in general dimensions $d\ge1$. 
	
	Let   $N\ge2$ and fix a Hartree dynamic $\varphi_t,\,t\ge0$ evolving according to \eqref{Hartree} with $\ell^2$-normalized initial state. For ease of notation, given $\xi_0\in \Fl$ and operator $\Ac$, we write 
	\begin{align}\label{def31}
		\br{\Ac}_t:=\br{\xi_t,\Ac\xi_t},\quad \xi_t:=\W_N(t;0)\xi_0.
	\end{align}
	For $R>r>0$, put
	\begin{align}
		\label{Mdef}
		M_{R}(t) :=& {\abs{\l}}\int_0^t \del{ \norm{\varphi_\tau}_{\ell^\infty(B_R)}+5\norm{\varphi_\tau}_{\ell^\infty(B_R)}^2}\,d\tau,\\
		\label{Edef}
		E_{R,r}(t):=&\  t\ \del{\frac{v}{ R-r }}^{n+1}{\sup_{0\le \tau\le t}\br{\cN_{B_{R+1}}}_\tau}  
		\notag\\+&
		{{{\frac{ \abs{\l}}{{N}}}\sup_{0\le \tau\le t}{\br{\del{\cN+1}^{2}}_\tau}}\int_0^t\del{ \norm{\varphi_\tau}_{\ell^4(B_R)}^4+{\norm{\varphi_\tau}_{\ell^\infty(B_R)}}}}\,d\tau
		.
	\end{align}
		%
		%
		%
		%
	%
	
	Finally, we define the velocity parameter
	\begin{align}
		\label{kappa}
		\kappa:=\sup_x\sum_y \abs{\Delta_{x,y}}\abs{x-y}{=2d}.
	\end{align}
	By Schur's test, it is easy to verify that the $1$-body momentum operator $i[-\Lap,\abs{x}]$ is bounded in norm by $\kappa$. 
	
	The main result of this section is the following:
	\begin{theorem}[Local fluctuation estimate]\label{thmLocFlucEst}
			%
		Let $d\ge1$.
		For every {$n\in\mathbb{N}_+$} and $v > \kappa$ {  there exists $C=C(n, v,d)>0$ such that for all $r>0$, $R\ge r+v$, and $0\le t \le \tfrac{R-r}{v}$, we have}
		\begin{align}
			\label{MVB2}
			&e^{-M_{R}(t)}\br{\N_{B_r}}_t \le  {(1+C(R-r)^{-1})\br{\N_{B_{R}}}_0+CE_{R,r}(t)}.
		\end{align}
		
	\end{theorem}

	\thmref{thmLocFlucEst}  is proved at the end of \secref{secPfThm3.1}. 	The proof  is based on the ASTLO (adiabatic spacetime localization observable) method introduced in \cite{FLS,FLSa} and developed in \cite{LRSZ,LRZ,Zhaa}. 
	
	The ASTLO, introduced in \secref{secDefASTLO}, is a class of propagation observables designed to track and control the evolution of number operators. In \secref{secDiffEst}, we prove a key differential inequality obeyed by the ASTLOs, which shows that ASTLOs are time-decaying up to a small remainder. In \secref{secIntEst} we use Gr\"onwall's inequality to bootstrap the differential inequality for ASTLO, leading to a convenient integral estimate. Finally, in \secref{secPfThm3.1}, combining the integral estimate with certain geometric properties relating ASTLOs to the excitation number operators associated to the propagation regions, we complete the proof of \thmref{thmLocFlucEst}.

	To appreciate  \eqref{MVB2}, one could compare with known results e.g.~\cite{Lee, DL}  in the continuum. Assume for the moment   that there exists $C$ independent of $N$ such that
	$$
	\br{\mathcal{N}}_t\le C.
	$$
	Substituting this global fluctuation estimate  back to bound the supremum in line \eqref{Edef}, we deduce from \eqref{MVB2} a local estimate of the form
	\begin{align}
		\label{locForm}
		\br{\N_{B_r}}_t\le C \text{ $\br{\N_{B_R}}_0+$ small remainder},
	\end{align}
	for any ball $B_{r} $ and for time up to $(R-r)/v $, as long as $R-r$ is large.  Later, in \secref{sec4}, we will turn this heuristic into a rigorous statement; see \corref{cor37} therein.

	\subsection{Definition of the ASTLOs}\label{secDefASTLO}

	Let $v>\kappa$ with $\kappa$ defined in \eqref{kappa} and set %
	$$
	v':=\frac12\del{\kappa+v}.
	$$For any function $f\in L^\infty(\Rb)$, $t\ge0$, and $R,\,s>0$, we define 
	\begin{equation}\label{ftsDef}
		f_{ts}(x)\equiv f(|x|_{ts}):=f\left(\frac{R-v't-|x|}{s}\right).
	\end{equation}
	The ASTLOs are then given by
	\begin{align}\label{propag-obs1}
		{		\N_{f,ts}:=\sum_{x\in \Zb^d}f_{ts}(x)n_x, \quad  f \in \E.}
	\end{align}
	Here, for $\eps:=v-v'>0$, we denote by $\cE$ the function class
	\begin{equation}\label{classE}
		\mathcal{E}\equiv \cE_\epsilon:=\left\{
		f\in C^{\infty} (\mathbb{R})\left\vert \begin{aligned} &f\geq0,\: f\equiv 0 \text{ on } (-\infty,\epsilon/2],f\equiv 1 \text{ on } [\epsilon,\infty) \\ &f'\geq 0,\:\sqrt{f'}\in C_c^{\infty} (\mathbb{R}),\:\supp f'\subset(\eps/2,\eps)
		\end{aligned}\right.\right\}.
	\end{equation}
	Essentially, elements in $\cE$ are {cutoff} functions with compactly supported derivatives, see \figref{figF} below.
	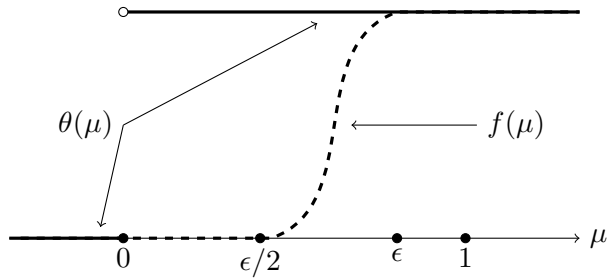
\begin{figure}[H]
		\centering
		\begin{tikzpicture}[scale=3]
			\draw [->] (-.5,0)--(2,0);
			\node [right] at (2,0) {$\mu$};
			\node [below] at (0,0) {$0$};
			\draw [fill] (0,0) circle [radius=0.02];
			
			\node [below] at (.6,0) {$\eps/2$};
			\draw [fill] (.6,0) circle [radius=0.02];
			
			\node [below] at (1.2,0) {$\eps$};
			\draw [fill] (1.2,0) circle [radius=0.02];
			
			\node [below] at (1.5,0) {$1 $};
			\draw [fill] (1.5,0) circle [radius=0.02];

			\draw [very thick] (-.5,0)--(0,0);
			\draw [very thick] (0,1)--(2,1);				
			\filldraw [fill=white] (0,1) circle [radius=0.02];
			
			\draw [dashed, very thick] (-.45,0)--(.65,0) [out=20, in=-160] to (1.2,1)--(2,1);

			\draw [->] (1.55,.5)--(1,.5);
			\node [right] at (1.55,.5) {$f(\mu)$};
			
			\draw [->] (0,.5)--(.85,.95);
			\draw [->] (0,.5)--(-.1,.05);
			\node [left] at (0,.5) {$\theta(\mu)$};
		\end{tikzpicture}
		\caption{A typical function $f\in \cE$ compared with the Heaviside function  $\theta(\mu)$.}\label{figF}
	\end{figure}
	Functions in $\cE$ are easy to construct. Indeed, for any $h\in C^\infty(\Rb)$ with $h\ge0$ and ${\supp h\subset (\eps/2, \eps)}$, let $f_1(\lambda):=\int_{-\infty}^\lambda h^2(s)ds$. Then  $f_1/\int_{-\infty}^\infty h^2(s)ds\in \cE$.
Similarly, one could also check that 
\begin{align}
	\label{E2}
	\text{If $f_1,f_2\in\cE$, then $f_1'+f_2'\le C f_3'$ for some $f_3\in\cE$.}
\end{align}
	
	In what follows we will reformulate \eqref{MVB2} in terms of certain estimates for the ASTLOs, and then  prove the corresponding ASTLO estimates to close the proof of \thmref{thmLocFlucEst}.

	\subsection{Differential inequality for ASTLO}\label{secDiffEst}
	Recall that the generator $\cL_N(t)$ of the fluctuation dynamics is given by \eqref{def:L}.
	Consequently, the evolution of any time-dependent propagation observables $\Phi(t)$ under the fluctuation dynamics is described by the Heisenberg equation
	\begin{align}
		\label{HeisEq}
		\od{}{t} \br{\N_{f,ts}}_t=\br{D\N_{f,ts}}_t, 
	\end{align}
	where $D$ is  the Heisenberg derivative
	\begin{align}\label{HeisDerDef}
		D\Phi(t):=\di_t\Phi(t)+i[\cL_N(t),\Phi(t)].
	\end{align}

	Fix $s>0\,,f\in\cE$, and define, according to    definition \eqref{propag-obs1}, the  ASTLOs 
	\begin{align}
		\label{ASTLO1}
		\Phi(t):=\N_{f,ts},\quad t\ge0.
	\end{align}
	The main result of this section is the following differential inequality:
	\begin{proposition}\label{prop36}
		Let $v'>\kappa$ and  $\delta:= v'-\kappa$. 
		Then for every $f\in\cE$, there exist $C=C(f,n)>0$ and, for $n\ge2$, functions $j_k\in\cE,\,k=2,\ldots,n$ (see \eqref{classE}) such that for all $s>0$ and   $\al,\,\beta,\,t\ge0$, the following operator inequality {holds on $\Fl$}:
		\begin{align}
			\label{DiffEst}
			D\Phi(t)	\le& -\delta s^{-1} {\N_{f',ts}} +  C \sum_{k=2}^n s^{-k} {\N_{j_k' ,ts}}   + \frac{ C(4d+1)}{s^{n+1}} \N_{B_{R+1}}\notag\\
			&+{\frac{\abs{\l}}{N}}\Bigl[ {\cN^{2}} {}\norm{\varphi_t}_{\ell^\al(B_R)}^{\al} + { \del{\cN+\cN^2}  {\norm{\varphi_t}_{\ell^\infty(B_R)}^{1-\beta}} }\Bigr]
			\notag\\
			&\qquad+{\abs{\l}}\del{4\norm{\varphi_t}_{\ell^\infty(B_R)}^2+{ \norm{\varphi_t}_{\ell^\infty(B_R)}^{1+\beta}}+ {\norm{\varphi_t}_{\ell^\infty(B_R)}^{6-\al}} }\Phi(t) .
		\end{align}
		The $k$-sum is dropped for $n=1$. 
	\end{proposition}
	\begin{proof}
		Within this proof we fix $s>0$ and $t\ge0$. All estimates below are independent of $s$ and $t$.
		
		1. 	
		In view of \eqref{HeisEq}, we first control the commutator $
		i[\cL_N(t),\Phi(t)]$. Choosing $h=f_{ts}\ge0$ in \lemref{lemma:commutator} and using definition \eqref{ASTLO1} for $\Phi(t)$,  we find
		\begin{align}
			\label{316}
			i \big[\mathcal{L}_N (t), \Phi(t)    \big]\leq&  i[\dG(-\Lap),\Phi(t)] \notag \\
			& +  {\abs{\l}}\sum_{x \in \Zb^d}  f_{ts}(x)  \;  \bigg( 4 \vert \varphi_t (x) \vert^2 n_x  + 2\frac{\vert \varphi_t (x) \vert^3}{\sqrt{N}} \mathcal{N} n_x^{1/2} +2 \frac{\vert \varphi_t (x) \vert}{\sqrt{N}} n_x (n_x+1)^{1/2} \bigg).
		\end{align}
		The leading term $ i[\dG(-\Lap),\Phi(t)]$ is well-understood and estimates for it are obtained e.g.~in \cite[eq.~(4.44)]{LRZ}. Here we record the result: for any $n=1,2,\ldots$,
		\begin{align}
			\label{ME1}
			{i[\dG(-\Lap),\Phi(t)]}\le& \kappa s^{-1} \N_{f',ts} +  C_{f,n} \sum_{k=2}^n s^{-k} \N_{j_k' ,ts}   +\frac{ C_{f,n}}{s^{n+1}}P,\\
			P:=&  \  \N_{B_{R}}+\sum_{x\in B_{R}}\sum_{y\in\Zb^d}\abs{\Delta_{x,y}}|x-y|^{n+1}{ n_y} . 
		\end{align}
		With the Laplacian matrix given by \eqref{def:Laplace}, the remainder $P$ can be bounded as 
		\begin{align}
			\label{Pest}
			P\le \N_{B_R}+2d\,\N_{B_{R+1}}.
		\end{align}
		Inserting this into \eqref{ME1} yields
		\begin{align}
			\label{ME1'}
			{i[\dG(-\Lap),\Phi(t)]}\le& \kappa s^{-1} \N_{f',ts} +  C_{f,n} \sum_{k=2}^n s^{-k} \N_{j_k' ,ts}   +\frac{ C_{f,n}(4d+1)}{s^{n+1}}\N_{B_{R+1}}.
		\end{align}

		2.	Next, we bound the remaining three terms in line \eqref{316}.  A  key observation is  that by definition \eqref{propag-obs1}, for any $f\in\cE$ and $s>0$, the function $f_{ts}$ satisfies %
		\begin{align}
			\label{fSupp}
			\supp f_{ts}\subset B_R,\quad t\ge0. 
		\end{align} Owning to \eqref{fSupp}, it suffices to take the various summations below over the ball $B_R$ only. 
		
		2.1. For the first  third term in line \eqref{316}, we use H\"older's inequality to obtain 
		%
		\begin{align}
			\label{319}
			&\sum_{x \in \Zb^d}  f_{ts}(x)  4 \vert \varphi_t (x) \vert^2 n_x   \le  {4\norm{\varphi_t}_{\ell^\infty(B_R)}^2 }\Phi(t).
		\end{align}
		
		2.2. For the second term in line \eqref{316}, we apply successively H\"older's, Cauchy-Schwartz, and Young's inequalities to obtain, for any $\al\ge0$,
		%

		\begin{align}
			\label{323}
			& { \sum_{x \in \Zb^d}  f_{ts}(x) 	2\frac{\vert \varphi_t (x) \vert^3}{\sqrt{N}} \mathcal{N} n_x^{1/2}}\notag\\
			\le&2\frac{\norm{\varphi_t}_{\ell^\infty(B_R)}^{3-\frac\al2}}{\sqrt{N}} \sum_{x \in B_R}  f_{ts}(x)\abs{\varphi_t(x)}^{\frac\al 2} {{{\cN}}  n_x^{1/2}}\notag\\
			\le & \sum_{x \in B_R}  2\sbr{{\norm{\varphi_t}_{\ell^\infty(B_R)}^{3-\frac{\al}{2}}}\del{f_{ts}(x)n_x}^{1/2}}\sbr{\frac{1}{\sqrt{N}}\del{f_{ts}(x)\abs{\varphi_t(x)}^{\al} {\cN^2} }^{1/2}}\notag\\
			\le& {\norm{\varphi_t}_{\ell^\infty(B_R)}^{6-\al}}{\Phi(t)}+\frac{\cN^{2}}{N}   \sum_{x \in B_R} f_{ts}(x)\abs{\varphi_t(x)}^{\al}. 		
		\end{align}
		Since  $f_{ts}(x)\le1$ for all $x$ (see \eqref{ftsDef}), we have $${  \sum_{x \in B_R} f_{ts}(x)\abs{\varphi_t(x)}^{\al}\le \norm{\varphi_t}_{\ell^\al(B_R)}^\al 		.}$$
		Plugging this back to \eqref{323},
		we conclude that 
		\begin{align}
			\label{323'}
			{ \sum_{x \in \Zb^d}  f_{ts}(x) 	2\frac{\vert \varphi_t (x) \vert^3}{\sqrt{N}} \mathcal{N} n_x^{1/2}}\le {\norm{\varphi_t}_{\ell^\infty(B_R)}^{6-\al}} {\Phi(t)}+\norm{\varphi_t}_{\ell^\al(B_R)}^{\al}\frac{\cN^{2}}{N }  {}.
		\end{align}
		
		2.3. For the third term in line \eqref{316}, we proceed similarly as in \eqref{323} to obtain, for any $\beta\ge0$, that
		\begin{align}
			\label{3230}
			&{ \sum_{x \in \Zb^d}  f_{ts}(x) 	2 \frac{\vert \varphi_t (x) \vert}{\sqrt{N}} n_x (n_x+1)^{1/2}}\notag\\
			\le&	2\frac{\norm{\varphi_t}_{\ell^\infty(B_R)}^{\frac12+\frac\beta 2}}{\sqrt{N}} \sum_{x \in B_R}  f_{ts}(x)\abs{\varphi_t(x)}^{\frac12-\frac \beta 2} { n_x (n_x+1)^{1/2}}\notag\\
			\le & \sum_{x \in B_R}  2\sbr{{\norm{\varphi_t}_{\ell^\infty(B_R)}^{\frac12+\frac\beta 2}}\del{f_{ts}(x)n_x}^{1/2}}\sbr{\frac{\norm{\varphi_t}_{\ell^\infty(B_R)}^{\frac12-\frac\beta 2}}{\sqrt{N}}\del{f_{ts}(x) n_x(n_x+1)}^{1/2}}\notag\\
			\le& {\norm{\varphi_t}_{\ell^\infty(B_R)}^{1+\beta}}{\Phi(t)}+\frac{\norm{\varphi_t}_{\ell^\infty(B_R)}^{1-\beta}}{N}   \sum_{x \in B_R} f_{ts}(x){n_x(n_x+1)}. 		
		\end{align}
		This implies, since $f_{ts}(x)\le1$, that
		\begin{align}
			\label{3231}
			{			\sum_{x \in \Zb^d}  f_{ts}(x) 	2 \frac{\vert \varphi_t (x) \vert}{\sqrt{N}} n_x (n_x+1)^{1/2}\le {\norm{\varphi_t}_{\ell^\infty(B_R)}^{1+\beta}} {\Phi(t)}+{\norm{\varphi_t}_{\ell^\infty(B_R)}^{1-\beta}}\frac{\cN+\cN^2}{N}  .}
		\end{align}
		
		3. Finally, the temporal derivative for $\Phi(t)$ is easily computed as 
		\begin{align} \label{eq:deriv}
			\di_t\Phi(t)=-s^{-1}v' \, \N_{f',ts}.
		\end{align}
		Plugging \eqref{eq:deriv},  \eqref{ME1'}, \eqref{319}, \eqref{323'}, and \eqref{3231} back to \eqref{HeisDerDef}, we
		conclude 
		\begin{align}
			\label{}
			D\Phi(t)=&\di_t\Phi(t)+i[\cL_N(t),\Phi(t)]\notag\\
			\le& - s^{-1}v' \, \N_{f',ts}\notag\\
			&+\kappa s^{-1} \N_{f',ts} +  C_{f,n} \sum_{k=2}^n s^{-k} \N_{j_k' ,ts}   +\frac{ C_{f,n}(4d+1)}{s^{n+1}}\N_{B_{R+1}}\notag\\
			&+{4\norm{\varphi_t}_{\ell^\infty(B_R)}^2 }\Phi(t)\notag\\
			&+{\norm{\varphi_t}_{\ell^\infty(B_R)}^{6-\al}}{\Phi(t)}+\frac{\cN^{2}}{N} {} \sum_{x \in B_R} f_{ts}(x)\abs{\varphi_t(x)}^{\al}\notag\\
			&+{\norm{\varphi_t}_{\ell^\infty(B_R)}^{1+\beta}} {\Phi(t)}+{\norm{\varphi_t}_{\ell^\infty(B_R)}^{1-\beta}}\frac{\cN+\cN^2}{N}  .
		\end{align}
		Recalling the definition $\delta = v'-\kappa>0$,  we conclude 
		the desired estimate \eqref{DiffEst}.

		%
	\end{proof}
	
	\subsection{Integral estimate for ASTLO}\label{secIntEst}
	The differential inequality \eqref{DiffEst} proved in the last subsection has a crucial   structure:  the second term has the same form as the leading, non-positive term.
	{Lifting \eqref{DiffEst} via Gr\"onwall's inequality}, we arrive at an integral estimate with a similar recursive structure; see \eqref{propag-est1} below. Later on, we will  iterate this integral estimate in the class $\cE$ to control the growth of the ASTLOs, see \secref{secPfThm3.1}. 
	
	For ease of notation, 
	for  any function $g\in L^\infty(\Rb)$  and $N_{g,ts}$ as in \eqref{propag-obs1}, we write
	\begin{align}
		\label{phitDef}
		g[t]:= \br{N_{g,ts}}_t.
	\end{align}
	The main result of this section is the following:
	\begin{proposition}\label{propIntRME}
		Assume \eqref{DiffEst} holds with $f\in\cE$, $\delta>0$. Then there exist ${C=C(f,n,d,\delta)>0}$ and, for $n\ge2$, functions $j_k\in\mathcal{E}$, $2\leq k\leq n$, such that for all $t\ge0$, {$s\ge1$,}
		\begin{align}\label{intRME}
			\int_0^t f'[\tau]\,d\tau
			\leq& C\del{ s f[0]+\sum_{k=2}^ns^{-k+2}j_k[0]+s^{-n}\Rem_n(s,t)},\\
			\Rem_n(s,t):=&  t\ {\sup_{\tau\le t}\br{\cN_{B_{R+1}}}_\tau + s^{n+1}\frac{\abs{\l}}{N}\sup_{\tau\le t}{\br{\del{\cN+1}^2}_\tau}}
			\int_0^t\del{\norm{\varphi_\tau }_{\ell^4(B_R)}^4+{\norm{\varphi_\tau}_{\ell^\infty(B_R)}}}\,d\tau.\label{RemDef}
		\end{align} 
		The sum in \eqref{intRME} is dropped for $n=1$.
	\end{proposition}
	
	\begin{proof}
		1. It follows from the Heisenberg equation \eqref{HeisEq}, the differential inequality \eqref{DiffEst} with $\al=4$, $\beta=0$,  the fact $\norm{\varphi_t}_{\ell^p}\le \norm{\varphi_t}_{\ell^q}$ for $1\le q\le p\le\infty$, and $\norm{\varphi_0}_{\ell^2}=1$ that the function $f[t]\equiv\br{\cN_{f,ts}}\equiv  \br{\Phi(t)}_t$ (see \eqref{phitDef}, \eqref{ASTLO1}) obeys  
		\begin{align}
			\label{338}
			\od{}{t}f[t]\le& A(t)f[t]+B(t),\\
			A(t)=&{\abs{\l} } \del{\norm{\varphi_t}_{\ell^\infty(B_R)} +5\norm{\varphi_t}_{\ell^\infty(B_R)}^2},\label{Adef}\\
			B(t)= &-\delta s^{-1} f'[t] +  C_{f,n} \sum_{k=2}^n s^{-k} j_k'[t]   + \frac{ C_{f,n,d}}{s^{n+1}}\label{Bdef} \br{\N_{B_{R+1}}}_t\\&+\frac{ \abs{\l}}{{N}}{\br{\del{\cN+1}^{2}}_t}\del{ {{\norm{\varphi_t}_{\ell^4(B_R)}^4}} +
				{\norm{\varphi_t}_{\ell^\infty(B_R)} } }
			.\notag
		\end{align}
		{Note that \propref{prop36} is indeed applicable since $\W_N(t;0)$   maps  $\Fl$  into itself according to \eqref{def:W} and \eqref{UNtDef}, so that $\xi_0\in\Fl$ implies $\xi_t\in\Fl$ for all $t$.}
		
		Applying Gr\"onwall's inequality to \eqref{338} yields
		\begin{align}
			\label{341}
			f[t]\le f[0]\exp(\int_0^tA(\tau)\,d\tau)+\int_0^tB(\tau)\exp(\int_\tau^t A(r)\,dr)\,d\tau.
		\end{align}
		Recalling definition $M_{R}(t)$ from \eqref{Mdef}, we see that $\int_\tau^t A\le  M_{R}(t)$ for any $0\le \tau\le t$, and so  \eqref{341} becomes
		\begin{align}
			\label{341'}
			e^{- M_{R}(t)}f[t]\le {f[0] +\int_0^tB(\tau) \,d\tau}.
		\end{align}

		2. Inserting expression \eqref{Bdef} into \eqref{341'}, we find
		\begin{align} \label{propag-est1} 
			&e^{-M_{R}(t)}f[t]+\delta s^{-1}\int_0^tf'[\tau]\, d\tau\notag\\\le& f[0]+C_{f,n} {\sum_{k=2}^n s^{-k}{\int_0^t} {j_k'[\tau]    }\,d\tau+\frac{C_{f,n,d}}{s^{n+1}}  \int_0^t \br{N_{B_{R+1}}}_\tau\,d\tau} \notag\\&+  \frac{\abs{\l}}{N}\int_0^t\br{\del{\cN+1}^2} _\tau\del{\norm{\varphi_\tau }_{\ell^4(B_R)}^4+{\norm{\varphi_\tau}_{\ell^\infty(B_R)} }}  \,d\tau.
		\end{align}
		Here and below, the $k$-sum is dropped for $n=1$. 	
		
		Applying H\"older's inequality   $\int_0^t \br{\cdot}_\tau\,d\tau\le t \sup_{\tau\le t}\br{\cdot}_\tau $ for the penultimate integral in \eqref{propag-est1}, 	we arrive at,  with $\Rem_n(s,t)$  given by \eqref{RemDef},
		\begin{align} \label{propag-est1'} 
			e^{-M_{R}(t)}f[t]+\delta s^{-1}\int_0^tf'[\tau]\, d\tau\le& f[0]+C_{f,n} \sum_{k=2}^n s^{-k}{\int_0^t} {j_k'[\tau]    }\,d\tau+\frac{C_{f,n,d}}{s^{n+1}} \Rem_n(s,t).
		\end{align}
		Since $\delta>0$, estimate \eqref{propag-est1} implies, after dropping $e^{-M_{R}(t)}f[t]\ge0$ and multiplying by $s\delta^{-1}>0$, that 
		\begin{align}
			\int_0^t f'[\tau] \,d\tau\, 
			\le&  \, 
			\frac{s}{\delta}f[0] + \frac{C_{f,n}}{\delta}\sum_{k=2}^n s^{-k+1}{\int_0^t} {j_k'[\tau]    }\,d\tau+\frac{C_{f,n,d}}{\delta s^{n}}\Rem_n(s,t). \label{propag-est2} 
		\end{align}

		3.	We now prove \eqref{intRME} by induction on $n$. For the base case   $n=1$, \eqref{propag-est2} gives the desired result. So let $n\ge2$. Assuming \eqref{intRME} holds for $1,2,\ldots,n-1$, we prove it for $n$. 
		
		Indeed, since the functions $j_k,\,k=2,\ldots,n$ all lie in the class $\cE$, we apply the induction hypothesis with $n-k+1$  to the term $\int_0^t j_k'[\tau]\, d\tau$, yielding
		\begin{align}
			\int_0^t j_k'[\tau] \,d\tau\, 
			\le&  \, 
			\frac{s}{\delta}j_2[0] + \frac{C_{j_k,n}}{\delta}\sum_{p=2}^{n-k+1} s^{-p+1}{\int_0^t} {j_{k,p}'[\tau]    }\,d\tau+\frac{C_{j_k,n,d}}{s^{n-k+1}}\Rem_{n-k+1}(s,t). \label{propag-est2'} 
		\end{align}
		Here the $p$-sum is dropped for $n=k$.  
		Plugging \eqref{propag-est2'} back to bound each term in the sum in \eqref{propag-est2} and
		using {property} \eqref{E2} for the function class $\cE$, we find some $  \tilde j_m\in\cE,\,m=3,\ldots n$ (dropped if $n=3$) depending only on $j_{k,p}$ s.th. 
		\begin{align}
			\int_0^t f'[\tau] \,d\tau
			\le& C_{f,\delta,n,d}  
			\Bigl(s f[0] + \sum_{k=2}^ns^{-k+2}j_k[0] + \sum_{m=3}^n s^{-k+1} \int_0^t   \tilde j_k'[\tau]\,d\tau   \notag\\
			& + s^{-n}(\Rem_n(s,t)+\sum_{k=2}^n\Rem_{n-k+1}(s,t))\Bigr). \label{propag-est33} 
		\end{align}
		Here the $m$-sum is dropped for $n=3$, in which case we are done. 
		
		Note that  there is exactly one less integral on the r.h.s.~of \eqref{propag-est33} compared to \eqref{propag-est2}. Thus, 
		after repeating the above procedure $n-1$ times,  no integral is present on the r.h.s.
		Finally, since $\Rem_{n-k}(s,t)\le \Rem_n(s,t)$ for $s\ge1$ by \eqref{RemDef}, 
		we arrive at the desired estimate \eqref{intRME} for $n\ge2$, $s\ge1$. This completes the induction and \propref{propIntRME} is proved.  
	\end{proof}
	
	\subsection{Proof of \thmref{thmLocFlucEst}}\label{secPfThm3.1}
	{	We  are now ready to prove \thmref{thmLocFlucEst}.} The main mechanism is the following consequence of the integral estimate \eqref{intRME}:
	\begin{proposition}\label{propAest1}
		Assume \eqref{intRME} holds with $\delta>0$. Then there exist ${C=C(f,n,d,\delta)>0}$ and, for $n\ge2$, functions $j_k,\,j_{k,p}\in\mathcal{E}$, $2\leq k\leq n$, $2\le p \le n-k+1$, such that for all $t\ge0$, $s\ge1$,
		\begin{align}
			\label{Aest1}
			e^{-M_{R}(t)}	f[t]\le f[0]+C\sum_{k=2}^ns^{-k+1}j_k[0]+C\sum_{k=2}^n\sum_{p=2}^{n-k+1}{s^{-p-k+2}j_{k,p}[0]}+Cs^{-n-1}\Rem_n(s,t).
		\end{align}
		Here the $k$-sum is dropped for $n=1$ and the $p$-sum is dropped for $k=n$. 
	\end{proposition}
	\begin{remark}
		Not that the double sum in the r.h.s.~of \eqref{Aest1}, if present, is of the order $\mathcal O(s^{-2})$. 
	\end{remark}
	
	\begin{proof}[Proof of \propref{propAest1}]

		Dropping the second one in the l.h.s.~of \eqref{propag-est1'} (which is non-negative) and using definition \eqref{RemDef}, we find that there exist $C=C(f, n,d)>0$ and, for $n\ge2$, functions $j_k\in\cE$ s.th. 
		\begin{align} \label{propEst10}
			&e^{-M_{R}(t)}f[t]\le f[0] +C \sum_{k=2}^n s^{-k}\int_0^t j_k'[\tau]\, d\tau + C s^{-n-1}\Rem_n(s,t). 
		\end{align}
		If $n=1$ then the sum in  \eqref{propEst10} is dropped and we conclude  the desired result. If $n\ge2$, then we apply \eqref{intRME} to estimate the integrated terms. Indeed, for each $k=2,\ldots,n$, we apply \propref{propIntRME} with $n-k$ to get, for some $C=C(f,n,d,v)>0$,
		\begin{align}
			\label{jkpEst}
			\int_0^t j_k'[\tau]\, d\tau \le C \del{sj_k[0]+\sum_{p=2}^{n-k+1}s^{-p+2}j_{k,p}[0]+s^{-(n-k+1)}\Rem_{n-k}(s,t) },
		\end{align}
		where the sum is dropped if $n-k\le 1$. Inserting \eqref{jkpEst} into \eqref{propEst10} and using that $\Rem_{n-k}(s,t)\le \Rem_n(s,t)$ for $s\ge1$ (again, see \eqref{RemDef}), we conclude that
		\begin{align}
 			\label{}
			e^{-M_{R}(t)}	f[t]\le f[0]+C\sum_{k=2}^n\del{s^{-k+1}j_k[0]+\sum_{p=2}^{n-k+1}s^{-k-p+2}j_{k,p}[0]}+Cs^{-n-1}\Rem_n(s,t).
		\end{align}
		From here	we conclude the desired estimate \eqref{Aest1} for $n\ge2$. 
	\end{proof}

	Finally, we use estimate \eqref{Aest1} for ASTLOs to conclude the proof of  \thmref{thmLocFlucEst}. 
	
	\begin{proof}[Proof of \thmref{thmLocFlucEst}]
		Given $v>\kappa+2\delta $, we take $v'=\frac12(\kappa+v)>\kappa+\delta$, $\eps=v-v'>0$ in \eqref{ftsDef}, \eqref{classE}, respectively. 
		We observe that for any function $f\in\cE$,  $r>0$, $R\ge r+v$, and 
		\begin{align}
			\label{stCond}
			{s=\frac{R-r}{v}},
		\end{align}
		there hold
		\begin{align}
			\label{geoEst}
			f_{0s}(x)\le \1_{B_R}(x)	,\quad 	\1_{B_r}(x)\le f_{ts}(x)\quad( 0< t\le s).
		\end{align}
		
		These inequalities are straightforward consequence of the geometric properties of the ASTLOs. See also \cite[Sect.~4.2]{LRZ} for details. Since the map $g\mapsto \dG(g)$ is positivity-preserving, \eqref{geoEst} lift to the following  operator inequalities:
		\begin{align}\label{GeoEst}
			\N_{f,0s}  \le  \N_{B_R} ,\quad 
			\N_{B_r}  \le   \N_{f,ts} \quad( 0\le  t\le s).
		\end{align}

		Now fix any $f\in\cE$. Since $v'-\kappa>\delta$ and $s\ge1$, applying successively \propref{prop36}, \propref{propIntRME}, and \propref{propAest1} yields some $C>0$ and functions $j_k,\,j_{k,p}\in\mathcal{E}$, $2\leq k\leq n$, $2\le p \le n-k+1$ depending only on $f,n,d,\delta$ such that \eqref{Aest1} holds.
		
		Since $f,\,j_k,$ and $j_{k,p}$ are all functions in the class $\cE$, they all satisfy estimates \eqref{GeoEst}. Therefore, applying \eqref{GeoEst} to bound the corresponding terms in \eqref{Aest1} from above and below, we conclude that, under condition  \eqref{stCond},
		\begin{align}			\label{MVBout}
			e^{-M_{R}(t)}	\big\lan N_{B_r} \big\ran_t \le(1+Cs^{-1}) \big\langle N_{B_R} \big\rangle_0 +  C s^{-n-1}  \Rem_n(s,t).
		\end{align}
		Comparing definitions \eqref{RemDef} for $\Rem_n(\cdot)$  and \eqref{Edef} for $E_{R,r}(t)$, we conclude that  
		\begin{align}
			\label{}
			s^{-n-1}\Rem_n(s,t) = E_{R,r}(t) \quad \text{for}  \quad s=\frac{R-r}{v}. 
		\end{align}
		This identity, together with \eqref{MVBout}, yields the desired estimate \eqref{MVB2}. 
		
		This completes the proof of \thmref{thmLocFlucEst}. \end{proof}
	\begin{remark}
By 		 \propref{prop36},   the conclusion of \thmref{thmLocFlucEst} remains valid if  \eqref{Edef} is  replaced by, with any $0\le\al\le6,\,0\le \beta\le 1$,
		\begin{align}
			M_{R,\al,\beta}(t) :=& \abs{\l}\int_0^t \del{\norm{\varphi_\tau}_{\ell^\infty(B_R)}^{6-\al}+\norm{\varphi_\tau}_{\ell^\infty(B_R)}^{1+\beta}+4\norm{\varphi_\tau}_{\ell^\infty(B_R)}^2}\,d\tau,\\
			\label{EabDef}
			\tilde E_{R,r,\al,\beta}(t)=&\  t\ \del{\frac{v}{ R-r }}^{n+1}{\sup_{0\le \tau\le t}\br{\cN_{B_{R+1}}}_\tau}  
			\notag\\+&
			{{\frac{ \abs{\l}}{{N}}\sup_{0\le \tau\le t}{\br{\cN+\cN^{2}}_\tau}}\int_0^t\del{ \norm{\varphi_\tau}_{\ell^\al(B_R)}^\al+{\norm{\varphi_\tau}_{\ell^\infty(B_R)}^{1-\beta}}}}\,d\tau.
		\end{align}

	\end{remark}

	\section{Global fluctuation estimate}\label{sec4}
	In this section, we   prove estimates on the global fluctuation $\br{\cN}_t$ under condition \ref{disCond}. 
	In \secref{secPfCor37}, we bootstrap \thmref{thm35}   with estimate \eqref{MVB2} to obtain an improved local fluctuation bound; see \corref{cor37}. In \secref{sec43}, we verify condition \ref{disCond} for the condensate dynamic \eqref{Hartree} for small $\abs{\l}$ and $\ell^1$ initial data. 
	

	Our main result is the following:
	\begin{theorem}[Global estimate]\label{thm35}
		Let $d\ge3$ and let condition \ref{disCond} hold.  
		Then, for $j=1,2,3$,	there exists a constant $C>0$ depending only on $j$,   the constant $c$ in \eqref{dis}, and the $\ell^1$ norm of $\varphi_0$, such that
		\begin{align}
			\label{Nest}
			\br{(\cN+1)^j}_t  \le& C\,\br{(\cN+1)^j}_0.
		\end{align}

			%
	\end{theorem}
	This theorem is proved shortly. The basic idea is to first bound $\br{\cN+1}_t$ generously in terms of a given solution $\varphi_t$ to the Hartree equation \eqref{Hartree}, and then get rid of the dependence on $\varphi_t$ through   dispersive estimate   \eqref{dis}.
	
	
	\begin{remark}\label{remScale}
		Comparing with existing result in the continuum \cite{Lee, DL}, where uniform-in-time estimates for the global fluctuation are obtained in $d=3$, the caveat in our case is that the dispersive estimate for the free propagator $e^{it\Lap}$  on lattice is strictly weaker than the corresponding estimates in the continuum; see \eqref{DisEst1}. 
	\end{remark}

	To shorten notation, for any operator $\mathcal{A}$ on $\mathcal{F}_{\perp \varphi_t}^{\leq N}$ and state $\mathcal{U}_{N,s} \psi_{N,s} \in \mathcal{F}_{\perp \varphi_s}^{\leq N}$, we write
	\begin{align}
		\label{def:Ast}
		\langle \mathcal{A} \rangle_{(t;s)} =  \langle \mathcal{W}_{N}(t;s) \mathcal{U}_{N,s} \psi_{N,s}, \mathcal{A} \mathcal{W}_{N}(t;s) \mathcal{U}_{N,s} \psi_{N,s} \rangle \; . 
	\end{align}
	In particular, for $\xi_t:=\mathcal{U}_{N,t} \psi_{N,t}$ and 	$\br{\Ac}_t:=\br{\xi_t,\Ac\xi_t}$ defined in \eqref{def31}, we have
\begin{align}
	\label{}
	\br{\cN^+_X}_t\equiv \br{\cN^+_X}_{(t,0)},\quad X\subset\Zb^d.
\end{align}
	The proof of \thmref{thm35} is based on estimates for $\br{\cN}_{(t;s)}$ in terms of given condensate dynamics,  as stated in the next lemma:
	\begin{lemma}\label{lem23} Let $\varphi_t,\,t\ge0$ be a solution to the Hartree equation \eqref{Hartree} with initial data $\varphi_0 \in \ell^2$, $\norm{\varphi_0}_{\ell^2}\le1$.
		Then for $j=1,2,3$,  there exists  $C_* = C_* (j)>0$ such that 
		\begin{align}
			\langle (\mathcal{N} + 1)^j \rangle_{(t;s)} \leq \langle (\mathcal{N} + 1 )^j \rangle_{(s;s)} \; \exp \bigg( C_* \abs{\lambda} \int_s^t d\tau \| \varphi_\tau \|_{\ell^\infty} \bigg),\quad 0\le s\le t. \label{eq:moment-bound}
		\end{align}
	\end{lemma}	
This lemma is proved by standard arguments for moment bounds of $\mathcal{N}$ in the continuum (see for example \cite{DL,RS}) and relies on commutator estimates of the generator of the fluctuation dynamics $\mathcal{L}_N (t)$ with the number of particles $\mathcal{N}$ stated in \lemref{lemma:commutator}.   We  postpone the proof  of Lemma \ref{lem23} to Appendix \ref{secPfLem23} and first prove \thmref{thm35}.
	
	\begin{proof}[Proof of \thmref{thm35} assuming \lemref{lem23}]
		Let 
		\begin{align}
			G(t):=&\exp( C_*\abs{\l}\int_0^t d\tau \| \varphi_\tau \|_{\ell^\infty}  ),\label{Gdef}
		\end{align} 
		with the   constant $C_*>0$ from \lemref{lem23}. By \eqref{eq:moment-bound}, we have
		\begin{align}
			\label{Nest0}
			\sup_{0\le\tau\le t}\br{ (\cN+1)^j }_\tau\le \br{(\cN+1)^j}_0\ G(t).
		\end{align}
		Under condition \ref{disCond}, with  the constant  $c>0$ from the r.h.s.~of \eqref{dis}, we have $G(t)\le e^{C_*\abs{\l}c}$. Plugging this back to \eqref{Nest0} yields  \eqref{Nest}.
		
	\end{proof}
	
	\subsection{Improved local fluctuation estimate}\label{secPfCor37}

	Combining \thmref{thm35}   and \thmref{thmLocFlucEst}, we arrive at the following  improved local fluctuation estimate:
	\begin{corollary}\label{cor37}
		Let the assumption of \thmref{thm35} hold. Then for every {$n\in\mathbb{N}_+$} and $v > \kappa$ (see \eqref{kappa}),   
		there exists  $C=C(n,d, v,\l,c, \norm{\varphi_0}_{\ell^1})>0$    such that for all $r>0$, $R\ge r+v$, and $0\le t \le \tfrac{R-r}{v}$, we have 
		
			\begin{align}
				\label{MVB2sim}
				\br{\N_{B_r}}_t \le& C(1+ (R-r)^{-1})\br{\N_{B_{R}}}_0\notag\\&+C\sbr{ \frac1N+(R-r)^{-n} }\br{(\cN+1)^2}_0 .
			\end{align}  
			%
	\end{corollary}
	\begin{remark}
		For $N$ and $R-r$ large, estimate  \eqref{MVB2sim} are indeed of the desired form \eqref{locForm}.
	\end{remark}
	
	
	\begin{proof}[Proof of \corref{cor37}]
		Let
		\begin{align}
			\label{Mdef'}
			M(t)  :=& \abs{\l}\int_0^t \del{ \norm{\varphi_\tau}_{\ell^\infty }+5\norm{\varphi_\tau}_{\ell^\infty }^2}\,d\tau ,\\
			\label{Edef'}
			E(t):=&   { \frac{\abs{\l}}{N}     } \int_0^t  {\norm{\varphi_\tau}^4_{\ell^4  }+\norm{\varphi_\tau}_{\ell^\infty}} \,d\tau.
		\end{align}
		Comparing definitions \eqref{Mdef}--\eqref{Edef} and \eqref{Mdef'}--\eqref{Edef'}, the local fluctuation estimate \eqref{MVB2} becomes, 
		for all $0\le t \le \tfrac{R-r}{v}$,
		\begin{align}\label{100}
			e^{-M(t)}\br{\N_{B_r}}_t \le&   (1+C(R-r)^{-1})\br{\N_{B_{R}}}_0\notag\\&+  \sbr{E(t)+(R-r)^{-n}}\sbr{ \sup_{0\le \tau\le t}\br{(\cN+1)^2}_\tau }.
		\end{align}
		The   factor $\sup_{0\le \tau\le t}\br{(\cN+1)^2}_\tau$ in \eqref{100} is handled by \thmref{thm35}, so it remains to bound the functions $M(t)$ and $E(t)$.
		
		Since $\norm{\varphi_t}_{\ell^2}$ is conserved under the Hartree equation \eqref{Hartree}, we have under condition \ref{disCond} that 
		\begin{align}\label{44Est}
			\int_0^t	{\norm{\varphi _s }_{\ell^4}^4 }\,ds\le& \int_0^t		\norm{\varphi_s}_{\ell^\infty}^2{\norm{\varphi _s }_{\ell^2}^2 }\,ds\le c\norm{\varphi_0}_{\ell^1}.
		\end{align}
		It follows from estimates \eqref{44Est},  \eqref{dis}, and  the  definitions \eqref{Mdef'}--\eqref{Edef'} that
			\begin{align}
				\label{}
				\label{MEst'}
				M(t)   \le& 6c\norm{\varphi_0}_{\ell^1}\abs{\l},\\
				\label{EEst'}
				E(t)   \le&     2c\norm{\varphi_0}_{\ell^1}\frac{\abs{\l}}{N}.
		\end{align}
		Plugging \eqref{MEst'}--\eqref{EEst'} and \eqref{Nest} back to \eqref{100} yields the desired result. 
	\end{proof}

	\subsection{Sufficient conditions for  \ref{disCond}}\label{sec43}
	In this section, we present explicit sufficient conditions for the dispersive estimate \eqref{dis} to be valid. More precisely, we verify Examples \ref{ex1}--\ref{ex2}. 
	\begin{proposition} \label{prop44}
		Take $\varphi_0\in\ell^1$ with $\norm{\varphi_0}_{\ell^2}=1$.
		\begin{enumerate}
			\item  
			Let $d\ge4$. Then there exist constants $\l_0,\,c>0$ depending only on $d$ s.th.~if $\abs{\l}\le\l_0$, then \eqref{dis} holds for all $t\ge0$.
			\item 
			Let $d=3$. Then there exist absolute constants $\l_0,\,c>0$ s.th.~if $\abs{\l}\le\l_0$ and 	$T:= { {e^{1/\sqrt{\abs{\l}}}-1}}$, then \eqref{dis} holds for all $t\le T$.
		\end{enumerate}
	\end{proposition}
	\begin{proof}
		We use the Strichartz estimates for the Hartree states $\varphi_t$ evolving according to \eqref{Hartree},  proved in \thmref{thmA1}. Indeed,   by \eqref{phiEst1},  for every $d\ge3$, there exist
		$\l_0,\,C>0$ depending only on $d$ such that for all $\abs{\l}\le\l_0$ and $\varphi_0\in\ell^1$ with $\norm{\varphi_0}_{\ell^2}=1$, 
		\begin{align}
			\label{aaaa}
			\norm{\varphi_t}_\infty\le \frac{C}{(1+t)^{d/3}}\norm{\varphi_0}_{1},
		\end{align}
		which is valid for all $t\ge0$ for $d\ge 4$, and for $0\le t \le T:=e^{1/\sqrt{\abs{\l}}}-1$ for $d=3$. Integrating \eqref{aaaa} in time yields
		\begin{align}\label{intEst}
			{\int_0^t\norm{\varphi_\tau }_{\ell^\infty}\,d\tau}\le& 
			\left\{\begin{aligned}
				&C\norm{\varphi_0}_{\ell^1},\quad t\ge0,\ d\ge4,\\
				&	\tfrac{C}{{\sqrt{\abs{\l}}}}\norm{\varphi_0}_{\ell^1} ,\quad 0\le t\le T,\ d=3. 
			\end{aligned}\right.
		\end{align}
		The desired results  follow from here. 
		
		%
		%
		%

	\end{proof}


	\section{Applications of \thmref{thmLocFlucEst} and  Proof of \thmref{thm:nloc}}\label{sec5}

	To render \thmref{thmLocFlucEst} useful, we seek   sufficient conditions, which are in some sense `local in space', that make the constants $M_{R}(t)$, $E_{R,r}(t)$ in  \eqref{Mdef}--\eqref{Edef} small for times up to $t\le (R-r)/v$. In particular, using such conditions together with  \thmref{thmLocFlucEst} and the global fluctuation bounds, we prove \thmref{thm:nloc}.
	

	
	The main result in this section is the following:
	\begin{theorem}
		\label{thmBallFluc}

		Assume the initial condensate $\varphi_0$ satisfies, for some $r,\,\rho>0 $,
		\begin{align}
			\label{phiLoc'}
			\varphi_0(x)=0,\quad \abs{x}\le r+2\rho.
		\end{align} 
		Assume also   the initial fluctuation satisfies %
		\begin{align}
			\label{xiLoc}
			n_x\xi_0=0  ,\quad \abs{x}\le r+\rho.
		\end{align} I.e.~the initial state for \eqref{HNdef} is purely factorized in $B_{r+\rho}$.

		Then for any  $v>\kappa$ (see \eqref{kappa}) and integer $n\ge1$, there exists $C=C(n,v)>0$ s.th.  
		\begin{align}
			\label{511}
			\br{\cN_{B_r}}_t\le C e^{C \rho^{-n}}\rho^{-n}\sup_{0\le\tau\le t}\br{\cN+1}_\tau,\quad t\le \rho/v.
		\end{align}
	\end{theorem}
	This theorem is proved in \secref{secPfThm61}. 
	\begin{proof}[Proof of \thmref{thm:nloc} assuming \thmref{thmBallFluc}]
		With definition \eqref{def31} and the identity \eqref{Nid} on the truncated Fock space $\mathcal{F}_{\perp \varphi_t}^{\leq N}$, we have, for all $t\ge0$,
		\begin{align}
			\langle \mathcal{N}_{B_r} \rangle_{t}  = \langle \psi_{N,t}, \mathcal{N}_{B_r}^+ (t) \psi_{N,t} \rangle .
		\end{align}
		Inserting the global fluctuation bound \eqref{MVB2sim} into \eqref{511} and using that $e^{C\rho^{-n}}\le\const$ for $\rho\ge1$, we arrive at the desired estimate \eqref{115}.
	\end{proof}
	
	
	{Note estimate \eqref{511} \textit{per se} does not require   assumption \ref{disCond}.  For the derivation of \eqref{511},  dispersive estimate \eqref{dis} is not used  and all we need  is the ballistic upper bound for NLS \eqref{Hartree} and proper initial geometric setup; see \thmref{thmNLSMVB}  below.   }

	As a result, for general condensate dynamics with $\ell^2$-normalized initial data, owning to the bound $ \norm{\varphi_t}_{\ell^\infty}\le \norm{\varphi_t}_{\ell^2}\equiv 1$ and \lemref{lem23}, inserting \eqref{eq:moment-bound} to bound the last factor in \eqref{511} yields a general local estimate
	\begin{align}
		\label{512}
		\br{\cN_{B_r}}_t\le \frac {C e^{C  \abs{\l}t }}{\rho^n}\br{\cN+1}_0, \quad t\le \rho/v,
	\end{align}
	for some  $C=C(n,v)>0$. 
	
	\subsection{Proof of \thmref{thmBallFluc}}\label{secPfThm61}
	To begin with,	we consider a very general class of  nonlinear Sch\"odinger equations on $\ell^2(\Zb^d)$, $d\ge1$ of the form
	\begin{align}
		\label{NLS}
		i\di_t \varphi = \cL \varphi+V(t)q+\abs{N(t,\varphi)} . 
	\end{align}
	Here $\cL$ is a self-adjoint operator with kernel $\cL_{x,y}$ satisfying, for some $n\ge1$,
	\begin{align}
		\label{Lcond}
		\sup_x\sum_y\abs{\cL_{x,y}} \abs{x-y}^{n+1}<\infty.
	\end{align}
	For example, if $L_{x,y}$ is given by the discrete Laplacian in \eqref{def:Laplace}, then \eqref{Lcond} holds uniformly for all $n$. 
	We assume the potential $V(t)$ is family of uniformly bounded multiplication operators,  and the  nonlinearity satisfies  $\abs{N(t,\varphi)}\le C_1$ for $t\in\Rb,\,\abs{\varphi}\le C_2$. 
	
	For $Y\subset \Zb^d$ and $\rho\ge0$, write $Y_\rho:=\Set{x\in\Zb^d:\abs{x-y}\le \rho, y\in Y}$ and $Y_\rho^\cp:=\Zb^d\setminus Y_\rho$. We have the following ballistic upper bound for \eqref{NLS}: 
	\begin{theorem}[C.f.~\cite{Zha}]\label{thmNLSMVB}
		For any $v>\sup_x\sum_y\abs{\cL_{x,y}} \abs{x-y}$ and integer $n\in\mathbb{N}_+$,  there exists $C=C(n,M,v)>0$ s.th. for any subset $Y\subset\Zb^d$, $\rho>0$, and any solution  $\varphi_t,\,t\ge0$ to \eqref{NLS} with initial state $\varphi_0$,
		\begin{align}
			\label{NLSpropEst}
			\norm{\varphi_t}_{\ell^2(Y_\rho^\cp)}^2\le (1+C\rho^{-1})\norm{\varphi_0}_{\ell^2(Y^\cp)}^2+ C\rho^{-n},\quad t\le \rho/v.
		\end{align}
	\end{theorem}
	\begin{proof}
		For $Y=B_r,\,r>0$, the exact same statement is proved in \cite[Thm.~1.4]{Zha}. Next, at the beginning of \cite[Sect.~2.2]{Zha}, it is explained how a straightforward adaption of the arguments extends the result for balls to general subsets $Y$,  leading to \eqref{NLSpropEst}.  
		Thus	we conclude the desired result \eqref{NLSpropEst}. 
	\end{proof}

	Assuming  the initial states $\varphi_0$, $\xi_0$ are both localized some  distance away from the origin, we  use \thmref{thmNLSMVB} together with \thmref{thmLocFlucEst} to show that the local fluctuation near the origin is suppressed up to time proportional to the localization distance.

	\begin{proof}[Proof of \thmref{thmBallFluc}]
		Let $R=r+\rho$. 
		We prove \eqref{511} in two steps: First we use \thmref{thmLocFlucEst} and \eqref{xiLoc} to bound the local fluctuation in terms of $\norm{\varphi_t}_{\ell^2(B_R)}$. Then we use \thmref{thmNLSMVB} and assumption \eqref{phiLoc'} to bound $\norm{\varphi_t}_{\ell^2(B_R)}$.
		
		1. We have by definitions \eqref{Mdef}--\eqref{Edef} that  for all $t\ge0$,
		\begin{align}
			\label{e59}
			M_{R}(t)\le& 6t \sup_{0\le\tau\le t}\norm{\varphi_\tau }_{\ell^2(B_R)},\\
			E_{R,r}(t)\le& \del{t(\rho/v)^{-(n+1)}+2t\sup_{0\le\tau\le t}\norm{\varphi_\tau}_{\ell^2(B_R)}}\sup_{0\le\tau\le t}\br{\cN+1}_\tau .\label{e510}
		\end{align}
		By the localization condition \eqref{xiLoc} on the initial fluctuation, the leading term involving $\br{\cN_{B_R}}_0$ in the r.h.s.~of \eqref{MVB2} vanishes. Therefore, by \thmref{thmLocFlucEst} and the bounds \eqref{e59}--\eqref{e510} above, we find some $C_1=C_1(n,d,v)>0$ s.th.~for $t\le \rho/v$,
		\begin{align}
			\label{513}
			\br{\cN_{B_r}}_t\le& C_1 \exp{ C_1\rho \sup_{0\le\tau\le t}\norm{\varphi_\tau }_{\ell^2(B_R)}} \del{\rho^{-n}+ \rho\sup_{0\le\tau\le t}\norm{\varphi_\tau}_{\ell^2(B_R)}}\sup_{0\le\tau\le t}\br{\cN+1}_\tau.
		\end{align}

		2. Now we apply  the localization estimate \eqref{NLSpropEst} to $Y=\supp\varphi_0$. By the localization condition \eqref{phiLoc'} on the initial condensate, we have $B_R\subset  Y_{\rho}^\cp$.  This yields, for any $v>\kappa$ and $m\ge1$, some $C_2=C_2(m,v)>0$ s.th.
		\begin{align}
			\label{514}
			\sup_{0\le\tau\le t}\norm{\varphi_\tau}_{\ell^2(B_{R}^\cp)}\le\sup_{0\le\tau\le t}\norm{\varphi_\tau}_{\ell^2(Y_\rho^\cp)}\le  C_2\rho^{-m},\quad t\le \rho/v.
		\end{align}
		Choosing $m\ge n+1$ and then inserting \eqref{514} to bound the corresponding quantities in \eqref{513}, we arrive at the desired estimate \eqref{511}.
	\end{proof}

	\section{Convergence of one-particle reduced densities and Proof of \thmref{thm:trace}} 	
	\label{sec:trace}

	In this Section we study the convergence of expectations of local observables and prove Theorem \ref{thm:trace}. The main tools are Theorem \ref{thm:nloc} and the lemma below:

	\begin{lemma}\label{lemma:trace-diff}
		{	Under the same assumptions as in Theorem \ref{thm:trace},}
		there exists a universal constant $C>0$ such that 
		\begin{multline}
			\big\vert \Tr  \del{O \big( \gamma_{\psi_{N,t}} -  N \vert \varphi_t \rangle \langle \varphi_t \vert \big)} \big\vert  \le  C  \|O \|_{\rm op} \\\times\bigg( \langle \mathcal{N}_{B_r} \rangle_{(t;0)}    +  \| \varphi_t \|_{\ell^2(B_r)}^2    \langle  \mathcal{N} \rangle_{(t;0)}  + \|  \varphi_t \|_{\ell^2( {B_r})} \int_{0}^t ds \;  {e^{ \abs{\l}\int_s^t d\tau \; \| \varphi_\tau \|_{\ell^\infty ( \mathbb{Z}^d)}^2} }\langle    \mathcal{N} + 1  \rangle_{(t;s)} \; \bigg) .
			\label{61}
		\end{multline}
	\end{lemma} 
	%
	%
	%

	We postpone the proof of Lemma \ref{lemma:trace-diff} to \secref{secPfLem61}  and first show the proof of Theorem \ref{thm:trace} from Lemma \ref{lemma:trace-diff}. Note here the r.h.s.~of \eqref{61} are given in terms of the local number of fluctuations $\mathcal{N}_{B_r}$ and the local condensate density $\| \varphi_t \|_{\ell^2( B_r)}$, which are small respectively owning to \thmref{thm:nloc} and \thmref{thmNLSMVB}.

	\begin{proof}[Proof of Theorem \ref{thm:trace} assuming \lemref{lemma:trace-diff}]

		
		To begin with, let $$(\tilde\rho,\tilde v ) :=(\rho/2, v/2).$$  Then for any $4d<v\le \rho$, we have $\tilde\rho\ge \tilde v>2d$ and so \thmref{thm:nloc} holds with this choice of $(\tilde\rho,\tilde v)$. We now bound the three terms in the r.h.s.~of \eqref{61} in order for $t\le \tilde\rho/\tilde v$. 
		

		For the first term, we have by \thmref{thm:nloc} that
		\begin{align}
			\label{62}
			\langle \mathcal{N}_{B_r} \rangle_t\le 	C\tilde\rho^{-n} \br{\del{\cN+1}^2}_0.
		\end{align}
		
		For the second term,  note that since 
		$\supp \varphi_0\subset \Set{\abs{x}\ge r+2\tilde \rho}$,
		  by \thmref{thmNLSMVB} we have
		\begin{align}
			\label{e63}
			\| \varphi_t \|_{\ell^2(B_r)}^2 \le C\tilde\rho^{-n}.
		\end{align}
		This, together with \thmref{thm35}, implies that 
		\begin{align}
			\label{63}
			\| \varphi_t \|_{\ell^2(B_r)}^2    \langle  \mathcal{N} \rangle_t\le C\tilde\rho^{-2n} \br{ {\cN+1} }_0.
		\end{align}

		For the third term, we first use Lemma \ref{lem23} to obtain
		\begin{align}
			\langle (\mathcal{N} + 1) \rangle_{(t;s)}
			\leq& \langle (\mathcal{N} + 1) \rangle_{0} \exp \bigg( C_* \abs{\lambda} \int_s^t d\tau \| \varphi_\tau \|_{\ell^\infty( \mathbb{Z}^d)} \bigg).
			\label{65}
		\end{align}
		Next, 
		by assumption \ref{disCond} and the fact that $\norm{\varphi_t}_{\ell^\infty}\le \norm{\varphi_t}_{\ell^2}\equiv 1$, there exists some  $C>0$ 
		depending only on the constant $c$ in \eqref{dis} and $\norm{\varphi_0}_{\ell^1}$ 
		such that
		\begin{align}
			\label{65''}
			\int_s^t d\tau \; \| \varphi_\tau \|_{\ell^\infty ( \mathbb{Z}^d)}^{j} \le C ,\quad 0\le s\le t, \ j\ge 1.
		\end{align}
		Combining \eqref{65}--\eqref{65''} and using \thmref{thmNLSMVB} again, we find that the third term in the r.h.s.~of \eqref{61} is bounded as 
		\begin{align}
			\label{65'}
			\|  \varphi_t \|_{\ell^2( {B_r})} \int_{0}^t ds \;  e^{ \int_s^t d\tau \; \| \varphi_\tau \|_{\ell^\infty ( \mathbb{Z}^d)}^2} \langle  ( \mathcal{N} + 1) \rangle_{(t;s)}\le C \tilde\rho^{-n}\langle \mathcal{N} + 1 \rangle_{0}.
		\end{align}
		
		Finally,  plugging    \eqref{62}, \eqref{63}, \eqref{65'} back to \eqref{lemma:trace-diff}, we find
		\begin{align}
			\big\vert \Tr  \del{O \big( \gamma_{\psi_{N,t}} -  N \vert \varphi_t \rangle \langle \varphi_t \vert \big)} \big\vert  \le& C   \|O \|_{\rm op}     \tilde\rho^{-n } \br{\del{\cN+1}^2}_0    \; . \label{66}
		\end{align}
With the definition \eqref{def31} and identity \eqref{Nid} on the truncated Fock space $\mathcal{F}_{\perp \varphi_t}^{\leq N}$, we have 
		$$
		\langle \mathcal{N}_{B_r} \rangle_t = \langle \psi_{N,t}, \mathcal{N}_{B_r}^+ (t) \psi_{N,t} \rangle .
		$$
		Since the initial state is purely factorized, it follows that
		$$\br{\del{\cN+1}^2}_0 =1. $$
		These, together with \eqref{66} and the fact that $\tilde \rho^{-n}  \le C_n\rho^{-n}$, yields the desired estimate \eqref{110}. 
	\end{proof} 
	\begin{remark}
		From the proof we see that, similar to \eqref{512}, without assuming condition \ref{disCond} but instead using \lemref{lem23}, estimate \eqref{110} becomes 
		\begin{align}
			\label{}
			\big\vert \Tr  \del{O \big( \gamma_{\psi_{N,t}} -  N \vert \varphi_t \rangle \langle \varphi_t \vert \big)} \big\vert  \le& \frac{C e^{C\abs{\l}t}  \|O \|_{\rm op}      }{\rho^n} \br{\del{\cN+1}^2}_0   , \label{110'}
		\end{align}
		where $C=C(n,v)>0$. 
	\end{remark}
	
	\subsection{Proof of Lemma \ref{lemma:trace-diff}}
	\label{secPfLem61}
	The proof of Lemma \ref{lemma:trace-diff} is based on the observation  that the remainder terms $\mathcal{R}_{N,t}^{(j)}$ of the generator of the fluctuation dynamics in \eqref{eq:L-sum} are sub-leading in the large particle limit. We will show that the flucutation can be approximated by an asymptotic dynamic with generator that is quadratic in modified creation and annihilation operators only, i.e. 
	\begin{align}
		\mathcal{W}_N (t;0) \approx \mathcal{W}_\infty (t;0), \quad \text{with} \quad i \partial_t \mathcal{W}_\infty (t;0) = \mathbb{H} \mathcal{W}_\infty (t;0) \; .  
	\end{align}
	For the proof we furthermore use that the action of the asymptotic dynamics $\mathcal{W}_\infty (t;0)$ on modified creation and annihilation operators is approximately known and given by 
	\begin{align}
		\label{eq:action-Winfty}
		\mathcal{W}_\infty^* (t;0) \big[ b^*(h) + b(h) \big] \mathcal{W}_\infty (t;0) ) \approx b^*( \mathcal{L}_{(t;0)} h) + b( \mathcal{L}_{(t;0)} h) 
	\end{align}
	for any $h \in \ell^2( \mathbb{Z}^d)$, where the operator $\mathcal{L}_{(t;s)}$ satisfies 
	\begin{align}
		\label{fEq}
		\quad i \partial_s \mathcal{L}_{(t;s)} = \bigg( - \Delta + \vert \varphi_s \vert^2 + \lambda \widetilde{K}_{1,s} (x) - \lambda \widetilde{K}_{2,s} J\bigg) \mathcal{L}_{(t;s)}, \quad\mathcal{L}_{(t;t)} = 1 .
	\end{align}
	Note that \eqref{eq:action-Winfty} is exact, if the modified creation and annihilation operators are replaced by the standard ones,  in the definition of the generator of the asymptotic dynamics $\mathcal{W}_\infty (t;0)$. 
	
	In the Lemma below we prove basic properties of the generator $\mathcal{L}_{(t;s)}$ that we use later in the proof of Lemma \ref{lemma:trace-diff}. For similar results in the continuum see \cite{RS22,RL23}.

	\begin{lemma} \label{lemma:f}
	Let $\varphi_t,\,t\ge0$ denote the solution to the Hartree equation \eqref{Hartree} with initial data $\varphi_0 \in \ell^2( \mathbb{Z}^d)$, and let $\mathcal{L}_{(t;s)}$ be given for $s \in [0,t]$ by \eqref{fEq}. Then, we have 
		\begin{align}
			\label{eq:estimate-f}
			{\|\mathcal{L}_{(t;s)} f \|_{\ell^2( \mathbb{Z}^d)}^2 \leq \|f \|_{\ell^2( \mathbb{Z}^d)}^2} \exp \bigg(  2 \vert \lambda \vert \int_s^t d\tau \; \| \varphi_\tau \|_{\ell^\infty ( \mathbb{Z}^d)}^2   \bigg)   \; . 
		\end{align}
	\end{lemma}
	
	\begin{proof}
		We compute 
		\begin{align}
			\|\mathcal{L}_{(t;s)} f  \|_{\ell^2( \mathbb{Z}^d)}^2 = \| f \|_{\ell^2( \mathbb{Z}^d)}^2 + 2 \lambda  i \Im \int_s^t d\tau \; \langle \mathcal{L}_{(t;\tau)} f, \widetilde{K}_{2,\tau} J \mathcal{L}_{(t;\tau)} f \rangle .
		\end{align}
		Since $\vert \widetilde{K}_{2,\tau} (x) \leq \vert \varphi_\tau ( x) \vert^2$ by definition \eqref{tkDef2}, we get 
		\begin{align}
			\|\mathcal{L}_{(t;s)} f  \|_{\ell^2( \mathbb{Z}^d)}^2 \leq \| f \|_{\ell^2( \mathbb{Z}^d)}^2  + 2 \vert \lambda\vert \int_s^t d\tau \| \varphi_\tau \|_{\ell^\infty ( \mathbb{Z}^d)}^2 \|\mathcal{L}_{(t;\tau)} f \|_2^2 
		\end{align}
		and we conclude \eqref{eq:estimate-f} by Gr\"onwall's inequality. 
	\end{proof}
	
	Recall the definition of $\br{\cdot}_{(t;s)}$ in  \eqref{def:Ast}. We are now ready to prove Lemma \ref{lemma:trace-diff}.

	\begin{proof}[Proof of Lemma \ref{lemma:trace-diff}]
		
		1. Denote by $\Om$ the vacuum in $\Fl$. By definition of the reduced one-particle density, we write with the definition of fluctuation dynamics 
		\begin{align}
			\Tr  \gamma_{\psi_{N,t}}   O 
			=& \langle \psi_{N,t}, \dG (O) \psi_{N,t} \rangle \notag \\
			=& \langle \Omega, \mathcal{W}_N^* (t;0) \mathcal{U}_{N,t} 
			\dG (O) \mathcal{U}_{N,t}^*  \mathcal{W}_N (t;0) \Omega \rangle \notag \\
			=& \langle \Omega, \mathcal{W}_N^* (t;0) \mathcal{U}_{N,t}^*  \; \sum_{x,y \in \Zb^d} O(x,y) \; a_x^*a_y  \; \mathcal{U}_{N,t} \mathcal{W}_N(t;0)  \Omega \rangle \notag \\
			=& \bigg\langle  \mathcal{U}_{N,t}^*  \; \sum_{x,y \in \Zb^d} O(x,y) \; a_x^*a_y  \; \mathcal{U}_{N,t} \bigg\rangle_{(t;0)} \; . 
		\end{align}
		From \eqref{def:b} and \eqref{eq:propU} it follows 
		\begin{align}
			\mathcal{U}_{N,t}^*  \; \sum_{x,y \in \Zb^d} O(x,y) \; a_x^*a_y  \; \mathcal{U}_{N,t}  = \dG (q_t O q_t) + \sqrt{N} \phi_+ (q_t O \varphi_t) + N \langle \varphi_t, O \varphi_t \rangle
		\end{align}
		where we introduced the notation, with the modified annhilation creation operators from \eqref{def:b},
		\begin{align}
			\phi_+ (h) = b^*(h) + b(h) , \quad \text{for} \quad h \in \ell^2( \Zb^d) \;  , 
		\end{align}
		and we arrive at 
		\begin{align}
			\label{eq:trace-diff2}
			\Tr  \gamma_{\psi_{N,t}}   O   - N \langle \varphi_t, O \varphi_t  \rangle  =\big\langle  \dG (q_t O q_t ) \big\rangle_{(t;0)}  + \sqrt{N}\big\langle  \phi_+ (q_t O \varphi_t) \big\rangle_{(t;0)}  \; . 
		\end{align}

		2. 	To prove Lemma \ref{lemma:trace-diff}, we need to control both terms in the r.h.s.~of \eqref{eq:trace-diff2}. We start with the first one and write, using that $q_t = 1- \vert \varphi_t \rangle \langle \varphi_t \vert$,
		\begin{align}
			\dG & (q_t O q_t ) \notag \\
			=& \sum_{x,y \in \Zb^d}\mathds{1}_{x \in B_r}   \bigg( O(x,y)  -   \overline{( O \varphi_t )}(x)\varphi_t (y)  -   \overline{\varphi_t (x)}  ( O \varphi_t )(y)+ \langle \varphi_t, O \varphi_t \rangle \varphi_t (x) \overline{\varphi_t} (y) \bigg)  \mathds{1}_{y \in B_r} \notag \\
			&+ \sum_{x,y \in \Zb^d}\mathds{1}_{x \in B_r^c} \bigg( \overline{\varphi_t (x)}  ( O \varphi_t )(y) + \langle \varphi_t, O \varphi_t \rangle \varphi_t (x) \overline{\varphi_t} (y) \bigg)  \mathds{1}_{y \in B_r}  \notag \\
			&+ \sum_{x,y \in \Zb^d}\mathds{1}_{x \in B_r} \bigg( \overline{( O \varphi_t )}(x)\varphi_t (y) + \langle \varphi_t, O \varphi_t \rangle \varphi_t (x) \overline{\varphi_t} (y) \bigg)  \mathds{1}_{y \in B_r^c} \notag \\
			&+\langle \varphi_t, O \varphi_t \rangle  \sum_{x,y \in \Zb^d}\mathds{1}_{x \in B_r^c}  \varphi_t (x) \overline{\varphi_t} (y)  \mathds{1}_{y \in B_r^c} .
		\end{align}
		We recall the notation $\mathcal{N}_{B_r} = \sum_{x \in B_r} a_x^* a_x$ and $\mathcal{N}_{B_r^c} = \sum_{x \in B_r^c} a_x^*a_x $, that allows to estimate the first term of the r.h.s. of \eqref{eq:trace-diff2} by 
		\begin{align}
			\big\vert\langle& \Omega, \mathcal{W}_N^*(t;0) \dG (q_t O q_t) \mathcal{W}_N (t;0) \Omega \rangle \vert \notag \\
			\leq& {\bigg( \|O \|_{\rm op} + 2\| O \varphi_t \|_{\ell^2( B_r)} \| \varphi_t \|_{\ell^2(B_r)}  +\vert \langle \varphi_t, O \varphi_t \rangle \vert \;  \| \varphi_t \|_{\ell^2(B_r)}^3 \bigg)} \langle \Omega, \mathcal{W}_N^*(t,0) \mathcal{N}_X \mathcal{W}_N (t;0) \Omega \rangle   \notag \\
			&+ 2\| \varphi_t \|_{\ell^2(B_r^c)} \| O \varphi_t \|_{\ell^2(B_r)}  \big( 1+ \vert \langle \varphi_t, O \varphi_t \rangle\vert  \big) \|\mathcal{N}_{B_r}^{1/2} \mathcal{W}_N(t;0) \Omega \| \; \|\mathcal{N}_{B_r^c}^{1/2} \mathcal{W}_N (t;0) \Omega \| \notag \\
			&+ \vert \langle \varphi_t, O \varphi_t \rangle \vert \; \| \varphi_t \|_{\ell^2 ( B_r^c)}^2 \langle \Omega, \mathcal{W}_N^* (t;0) \mathcal{N}_{B_r^c} \mathcal{W}_N (t,0) \Omega \rangle \; .
		\end{align}
		Since $\| \varphi_t \|_{\ell^2( B_r^c)} \leq \|\varphi_t \|_{\ell^2 ( \Zb^d)}$ and $\mathcal{N}_{B_r^c} \leq \mathcal{N}$, we furthermore have 
		\begin{align}
			\big\vert\langle \Omega, & \mathcal{W}_N^*(t;0) \dG (q_t O q_t) \mathcal{W}_N (t;0) \Omega \rangle \vert \notag \\
			\leq& C  \|O \|_{\rm op}  \langle \Omega, \mathcal{W}_N^*(t,0) \mathcal{N}_{B_r} \mathcal{W}_N (t;0) \Omega \rangle    + C  \| O \|_{\rm op} \| \varphi_t \|_{\ell^2(B_r)}^2    \langle \Omega, \mathcal{W}_N^* (t;0) \mathcal{N} \mathcal{W}_N (t,0) \Omega \rangle \; .
			\label{eq:trace-diff-bound1}
		\end{align}
		This bounds the first term  in the r.h.s.~of \eqref{eq:trace-diff2}.
		
		
		3.	For the second term of \eqref{eq:trace-diff2} a straight forward estimate in the same spirit as for the first term would yield a bound of $\mathcal{O}(\sqrt{N})$. To improve the scaling in $N$, we use that the flucutation dynamics effectively acts as an asymptotic Bogoliubov dynamics on the observable $\phi_+ (q_t O \varphi_t)$. For this we define the function 
		\begin{align}
			\label{def:f}
			f_{(t,s)} = \mathcal{L}_{(t;s)} q_t O \varphi_t,
		\end{align}
		where the generator $\mathcal{L}_{(t;s)}$ is defined by \eqref{fEq}. The goal is now to show that for the second term of the r.h.s. of \eqref{eq:trace-diff2} we have 
		\begin{align}
			\sqrt{N}  \langle \Omega,  \mathcal{W}_N^*(t;0) \phi_+ (q_t O \varphi_t ) \mathcal{W}_N (t;0) \Omega \rangle \approx  \sqrt{N}\langle \Omega, \phi_+ ( f_{(t,0)} ) \Omega \rangle,  
		\end{align}
		while the r.h.s.~vanishes as {$b(h) \Omega  =0$.}

		To be more precise, we have 
		\begin{align}
			\sqrt{N}  \langle \Omega,  &\mathcal{W}_N^*(t;0) \phi_+ (q_t O \varphi_t ) \mathcal{W}_N (t;0) \Omega \rangle \notag \\ 
			=& 	\sqrt{N}  \langle \Omega,  \mathcal{W}_N^*(t;0) \phi_+ (q_t O \varphi_t ) \mathcal{W}_N (t;0) \Omega \rangle - \sqrt{N}\langle \Omega, \phi_+ ( f_{(t,0)} ) \Omega \rangle \notag \\
			=& i\sqrt{N} \int_0^t ds \; \langle \Omega, \mathcal{W}_N (t;s) \bigg( \big[\mathbb{H}, \phi_+ ( f_{(t;s)} ) \big] -\phi_- ( i \partial_s f_{(t;s)} ) \bigg) \mathcal{W}_N (t;s) \Omega \rangle \notag \\
			&+ i \sqrt{N} \int_0^t ds \; \langle \Omega, \mathcal{W}_N^*(t;s) \big[ \mathcal{R}_{N,t}, \phi_+ ( f_{(t,s)} )\big] \mathcal{W}_N (t;s) \Omega \rangle  \label{eq:comm-L-Phi}
		\end{align}
		where we introduced the notation 
		\begin{align}
			i \phi_- (h ) = b(h) - b^*(h) \quad \text{for all} \quad h \in \ell^2( \mathbb{Z}^d) \; . 
		\end{align}
		Below we bound the two terms in the r.h.s.~of \eqref{eq:comm-L-Phi} respectively. 
		
		3.1. For the first term of the r.h.s. of \eqref{eq:comm-L-Phi}, we compute 
		\begin{align}
			\big[\mathbb{H}, \phi_+ ( f_{(t;s)} ) \big] -\phi_- ( i \partial_s f_{(t;s)} ) =& \frac{\lambda}{2} b^*\big( \widetilde{K}_{2,s} J f_{(t;s)} \big)   \frac{\mathcal{N}_+(t)}{N} - {\rm h.c.} \notag \\
			&+ \frac{\lambda}{2N} \sum_{x \in \mathbb{Z}^d} \widetilde{K}_{2,s} (x) a^*(f_{(t,s)} ) a^*_xb_x - {\rm h.c.} \notag \\
			&+ \frac{\lambda}{2N}\sum_{x \in \mathbb{Z}^d}  \widetilde{K}_{2,s} (x) b_x^*a_x^* a( f_{(t;s)} ) - {\rm h.c.} \label{eq:comm-H-Phi}
		\end{align}
		and estimate the terms separately. For the first, we get with $ \sup_{x \in \mathbb{Z}^d} \vert \widetilde{K}_{2,s} (x) \vert \leq C$  (see \eqref{tkDef2}) that
		\begin{align}
			\vert \langle \Omega,  &\mathcal{W}^*_N (t;s) b^*\big( \widetilde{K}_{2,s} J f_{(t;s)} \big)   \frac{\mathcal{N}^+(t)}{N} \mathcal{W}_N (t;s) \Omega \rangle \vert \notag \\
			\leq&  \frac{C}{N}\| b^*( \widetilde{K}_{2,s} J f_{(t;s)} )( \mathcal{N} + 1)^{1/4} \mathcal{W}_N (t;s) \Omega \| \; \|  \mathcal{N}^{3/4} \mathcal{W}_N (t;s) \Omega \| \notag \\
			\leq& \frac{C}{N} \| f_{(t;s)} \|_{\ell^2( \mathbb{Z}^d)} \langle \Omega, \mathcal{W}^*_N (t;s) ( \mathcal{N} + 1 )^{3/2} \mathcal{W}_N (t;s) \Omega \rangle
		\end{align}
		and the hermitian conjugate can be estimated similarly. Moreover, 
		\begin{align}
		&	\frac{1}{2N}  \vert \langle \Omega,   \mathcal{W}_N( t;s)  \sum_{x \in \mathbb{Z}^d} \widetilde{K}_{2,s} (x) a^* ( f_{(t;s)} ) a_x^*b_x \mathcal{W}_N (t;s) \Omega \rangle \vert \notag \\
			\leq&  \frac{1}{2N} \| \sum_{x \in \mathbb{Z}^d}   \widetilde{K}_{2,s} (x)  a^*_x b_x ( \mathcal{N} +1)^{-1/4}\mathcal{W}_N (t;s) \Omega \| \; \| a( f_{(t;s)} ) \mathcal{N}^{1/4}\mathcal{W}_N (t;s) \Omega \| \notag \\
			=&  \;  \frac{1}{2N}\| \sum_{x \in \mathbb{Z}^d} \widetilde{K}_{2,s} (x)   \frac{ \sqrt{N - \mathcal{N}^+ (t) +1}}{\sqrt{N}} a_x^* a_x ( \mathcal{N} +1)^{-1/4}\mathcal{W}_N (t;s) \Omega \| \; \| a( f_{(t;s)} ) \mathcal{N}^{1/4}\mathcal{W}_N (t;s) \Omega \| \notag \\ 
			\leq& \frac{C}{N} \| f_{(t;s)} \|_{\ell^2( \mathbb{Z}^d)} \|( \mathcal{N} + 1)^{3/4} \mathcal{W}_N (t;s) \Omega \|^2 \; 
		\end{align}
		and the hermitian conjugate as well as the remaining term of the r.h.s.~of \eqref{eq:comm-H-Phi} can be estimated similarly. Thus we arrive at the following bound for the first term in the r.h.s.~of \eqref{eq:comm-L-Phi}:
		\begin{align}
			\vert \langle \Omega,  & \; \mathcal{W}_N (t;s) \bigg( \big[\mathbb{H}, \phi_+ ( f_{(t;s)} ) \big] -\phi_- ( i \partial_s f_{(t;s)} ) \bigg) \mathcal{W}_N (t;s) \Omega \rangle \vert \notag \\
			\leq& {\frac{C\abs{\l}}{N}} \| f_{(t;s)} \|_{\ell^2( \mathbb{Z}^d)} \langle \Omega, \mathcal{W}^*_N (t;s) ( \mathcal{N} + 1 )^{3/2} \mathcal{W}_N (t;s) \Omega \rangle \;. \label{eq:comm-phi-H-end}
		\end{align}
		
		3.2. For the second term in the r.h.s.~of \eqref{eq:comm-L-Phi}, we recall the splitting {$\mathcal{R}_{N,s} = \sum_{j=1}^3 \mathcal{R}_{N,s}^{(j)}$ from \eqref{def:Ri}}
		and bound their contributions separately. 
		
		\textbf{Contribution from $\mathcal{R}_{N,s}^{(1)}$.}  We compute 
		\begin{align}
			\big[\mathcal{R}_{N,s}^{(1)}, \phi_+ (f_{(t;s)}) \big] =& {\frac{\lambda}{2} i \phi_- ( q_s \big[ \vert \varphi_s\vert^2 + \widetilde{K}_{1,s} - \mu_s \big] f_{(t;s)} \big) \frac{1 - \mathcal{N}^+ (t)}{N} }\notag \\
			&- \frac{\lambda}{N} \dG \big( q_s \big[ \vert \varphi_s \vert^2 + \widetilde K_{1,s} - \mu_s \big] q_s \big) \; i \phi_- ( f_{(t;s)} ) \notag \\
			&- {\rm h.c.}  
		\end{align}
		where we introduced the notation $i \phi_- (h) = b(h) - b^*(h)$. Since $\sup_{x \in \mathbb{Z}^d} \vert\vert \varphi_s (x) \vert^2 +\widetilde K_{1,s}(x) - \mu_s  \vert \leq C $ (see \eqref{tKdef} and \eqref{mutDef}),  we get 
		\begin{align}
			\vert \langle \Omega,  & \mathcal{W}^*_N (t;s) \big[\mathcal{R}_{N,s}^{(1)}, \phi_+ (f_{(t;s)}) \big]  \mathcal{W}_N (t;s) \Omega \rangle \vert \notag \\
			&\leq  {\frac{C\abs{\l}}{N}} \|f_{(t;s)} \|_{\ell^2( \mathbb{Z}^d)} \langle \Omega, \mathcal{W}^*_N (t;s) ( \mathcal{N} + 1)^{3/2} \mathcal{W}_N (t;s) \Omega \rangle \; .  \label{eq:comm-phi-R1-end}
		\end{align}
		
		\textbf{Contribution from $\mathcal{R}_{N,s}^{(2)}$.} Using that
		\begin{align}
			\big[ \mathcal{R}_{N,s}^{(2)}, \phi_+ (f_{(t;s)} )\big] =& \frac{\lambda}{\sqrt{N}} \sum_{x \in \mathbb{Z}^d} \varphi_s (x) \overline{f}_{(t;s)} (x) b^* (q_{s,x}) b(q_{s,x} ) \notag \\
			&- \frac{\lambda}{\sqrt{N}} \sum_{x \in \mathbb{Z}^d} \varphi_s (x) f_{(t;s)} (x) b(q_{s,x} ) b(q_{s,x}) \notag \\
			&- \frac{\lambda}{\sqrt{N}} \sum_{x \in \mathbb{Z}^d} \varphi_s (x) f_{(t;s)} (x) a^*(q_{s,x}) a(q_{s,x}) \bigg( 1- \frac{\mathcal{N}}{N} \bigg) \notag \\
			&- \frac{\lambda}{N^{3/2}} \sum_{x \in \mathbb{Z}^d} \varphi_s (x) a^*(q_{s,x}) a(q_{s,x}) a^*( f_{(t;s)} ) b(q_{s,x} ) \notag \\
			&-{\rm h.c.} \label{eq:comm-R2-Phi}
		\end{align}
		and  that $\sup_{x \in \mathbb{Z}^d} \vert f_{(t;s)}  (x) \vert^2 \leq  \|f_{(t;s)} \|_{\ell^2(\Zb^d)}^2$, we get via the Cauchy-Schwartz inequality that
		\begin{align}
			\frac{1}{\sqrt{N}}  &  \vert \langle \Omega, \mathcal{W}^*_N (t;s)   \sum_{x \in \mathbb{Z}^d} \varphi_s (x) \overline{f}_{(t;s)} (x) b^* (q_{s,x}) b(q_{s,x} ) \mathcal{W}_N (t;s) \Omega \rangle \vert \notag \\
			\leq& \frac{C}{\sqrt{N}} \bigg( \sum_{x \in \mathbb{Z}^d } \vert \varphi_s (x) \vert^2 \|b (q_{s,x} ) \mathcal{W}_{N} (t;s) \Omega \|^2 \bigg)^{1/2} \bigg( \sum_{x \in \mathbb{Z}^d } \vert f_{(t;s)} (x) \vert^2 \|b (q_{s,x} ) \mathcal{W}_{N} (t;s) \Omega \|^2 \bigg)^{1/2} \notag \\ 
			\leq& \frac{C}{\sqrt{N}} \| f_{(t;s)} \|_{\ell^2( \mathbb{Z}^d)} \langle \Omega, \mathcal{W}^*_N (t;s) ( \mathcal{N} + 1) \mathcal{W}_N (t;s) \Omega \rangle  \; .
		\end{align}
		Since $\| \varphi_s \|_{\ell^2( \mathbb{Z}^d)} \leq 1$, the second and third term of the r.h.s.~of \eqref{eq:comm-R2-Phi} can be estimated similarly. For the forth term of the r.h.s.~of \eqref{eq:comm-R2-Phi} we find 
		\begin{align}
			&\frac{1}{N^{3/2}}  \vert \langle \Omega, \mathcal{W}^*_N (t;s) \sum_{x \in \mathbb{Z}^d} \varphi_s (x) a^*( q_{s,x}) a(q_{s,x}) a^*(f_{(t;s)} ) b(q_{s,x} ) \mathcal{W}_N (t;s) \Omega \rangle \vert \notag \\
			\leq& \frac{1}{N^{3/2}}\bigg( \sum_{x \in \mathbb{Z}^d} \vert \varphi_s (x) \vert^2 \|a^*(q_{s,x} ) a(q_{s,x} ) \mathcal{W}_{N} (t;s) \Omega \|^2 \bigg)^{1/2} \bigg( \sum_{x \in \mathbb{Z}^d} \|a^*(f_{(t;s)}) b(q_{s;x} ) \mathcal{W}_N (t;s) \Omega \|^2 \bigg)^{1/2} \notag \\
			\leq& \frac{C}{\sqrt{N}} \|f_{(t;s)} \|_{\ell^2( \mathbb{Z}^d)} \langle \Omega, \mathcal{W}^*_N (t;s) ( \mathcal{N} + 1) \mathcal{W}_N (t;s) \Omega \rangle  \; . 
		\end{align}
		Thus, we arrive at 
		\begin{align}
			\vert \langle \Omega,  & \mathcal{W}_N (t;s) \big[\mathcal{R}_{N,s}^{(2)}, \phi_+ (f_{(t;s)}) \big]  \mathcal{W}_N (t;s) \Omega \rangle \vert \notag \\
			&\leq  \frac{C\abs{\l}}{\sqrt{N}} \|f_{(t;s)} \|_{\ell^2( \mathbb{Z}^d)} \langle \Omega, \mathcal{W}^*_N (t;s) ( \mathcal{N} + 1) \mathcal{W}_N (t;s) \Omega \rangle \; . \label{eq:comm-phi-R2-end}
		\end{align}
		
		\textbf{Contribution from $\mathcal{R}_{N,s}^{(3)}$.}
		It remains to control 
		\begin{align}
			\big[  \mathcal{R}_{N,s}^{(3)}, \phi_+ (f_{(t;s)} )\big] = \frac{\lambda}{N} \sum_{x \in \mathbb{Z}^d } \; \del{  \overline{f}_{(t;s)} (x) b_x^*a^*(q_{s,x}) a(q_{s,x}) - \rm{h.c.}}  .
		\end{align}
		We estimate this by
		\begin{align}
		  &  \big\vert \langle \mathcal{W}_N (t;s)\Omega,  \big[  \mathcal{R}_{N,s}^{(3)}, \phi_+ (f_{(t;s)} )\big] \mathcal{W}_N (t;s)\Omega \rangle \big\vert \notag \\
			\leq& \frac{2 \abs{\l}}{N}  \sum_{x \in \mathbb{Z}^d} \vert f_{(t;s)}(x) \vert \; \| b_x \mathcal{W}_N(t;s) \Omega \| \; \| a^*(q_{s,x})a(q_{s,x}) \mathcal{W}_N (t;s) \Omega \| \notag \\
			\leq& \frac{2 \abs{\l} }{N} \|f_{(t;s)} \|_{\ell^2( \mathbb{Z}^d)} \bigg( \sum_{x \in \mathbb{Z}^d} \| b_x \mathcal{W}_N(t;s) \Omega \|^2 \| a^*(q_{s,x})a(q_{s,x}) \mathcal{W}_N (t;s) \Omega \|^2 \big)^{1/2} \notag \\
			\leq& \frac{2 \abs{\l}}{N} \|f_{(t;s)} \|_{\ell^2(\mathbb{Z}^d)} \| (\mathcal{N}+1)^{1/2} \mathcal{W}_N (t;s) \Omega \| \; \| ( \mathcal{N}+1) \mathcal{W}_N (t;s) \Omega \| \;. \label{eq:comm-phi-R3-end}
		\end{align}
		Since $\cN\le N$ on $\Fl$ and $\norm{(\cN+1)^\al\W(t;s)\Om}^2= \langle \Omega, \mathcal{W}^*_N (t;s) ( \mathcal{N} + 1)^{2\al}  \mathcal{W}_N (t;s) \Omega \rangle$, we have
\begin{align}
	\label{}
			 &\frac{1}{N } \| (\mathcal{N}+1)^{1/2} \mathcal{W}_N (t;s) \Omega \| \; \| ( \mathcal{N}+1) \mathcal{W}_N (t;s) \Omega \|\notag\\=&\frac{1}{\sqrt{N}}\langle \Omega, \mathcal{W}^*_N (t;s) ( \mathcal{N} + 1)  \mathcal{W}_N (t;s) \Omega \rangle^{1/2} \langle \Omega, \mathcal{W}^*_N (t;s) \frac{ \mathcal{N} + 1}{N}  \mathcal{W}_N (t;s) \Omega \rangle^{1/2}\notag\\
			\le& \frac{1}{\sqrt{N}}\langle \Omega, \mathcal{W}^*_N (t;s) ( \mathcal{N} + 1)  \mathcal{W}_N (t;s) \Omega \rangle\notag.
\end{align}
	Plugging this back to  \eqref{eq:comm-phi-R3-end} yields
		\begin{align}
					  &  \big\vert \langle \mathcal{W}_N (t;s)\Omega,  \big[  \mathcal{R}_{N,s}^{(3)}, \phi_+ (f_{(t;s)} )\big] \mathcal{W}_N (t;s)\Omega \rangle \big\vert \notag \\
			\leq& \frac{2 \abs{\l}}{\sqrt N} \|f_{(t;s)} \|_{\ell^2(\mathbb{Z}^d)} \langle \Omega, \mathcal{W}^*_N (t;s) ( \mathcal{N} + 1)  \mathcal{W}_N (t;s) \Omega \rangle \label{eq:comm-phi-R3-end'}.
		\end{align}

		Combining \eqref{eq:comm-phi-R1-end}, \eqref{eq:comm-phi-R2-end} and \eqref{eq:comm-phi-R3-end'}, we conclude 
		\begin{align}
			\label{2bdd}
		 &\abs{ \langle \Omega, \mathcal{W}_N^*(t;s) \big[ \mathcal{R}_{N,t}, \phi_+ ( f_{(t,s)} )\big] \mathcal{W}_N (t;s) \Omega \rangle}\notag\\\le&  \frac{C\abs{\l}}{\sqrt{N}} \|f_{(t;s)} \|_{\ell^2( \mathbb{Z}^d)} \langle \Omega, \mathcal{W}^*_N (t;s) ( \mathcal{N} + 1)  \mathcal{W}_N (t;s) \Omega \rangle \; .
		\end{align}
		This bounds the second term in the r.h.s.~of \eqref{eq:comm-L-Phi}.
		
		4.	Summarizing \eqref{eq:comm-phi-H-end} and \eqref{2bdd}, we finally arrive at the following bound for  \eqref{eq:comm-L-Phi}:
		\begin{align}\label{638}
			\big\vert \langle\Omega,  &\mathcal{W}_N^*(t;0) \phi_+ ( q_t O \varphi_t) \mathcal{W}_N (t;0) \Omega \rangle \big\vert  \notag \\ 
			&\leq  \frac{C\abs{\l}}{\sqrt{N}} \int_0^t \|f_{(t;s)} \|_{\ell^2( \mathbb{Z}^d)} \langle \Omega, \mathcal{W}^*_N (t;s) ( \mathcal{N} + 1)  \mathcal{W}_N (t;s) \Omega \rangle \; . 
		\end{align}
		Furthermore, from Lemma \ref{lemma:f} we  get, recalling $\| f_{(t;t)} \|_{\ell^2( \mathbb{Z}^d)} = \| q_t O \varphi_t \|_{\ell^2( \mathbb{Z}^d)} \leq \|O \|_{\rm op} \|\varphi_t \|_{\ell^2(  B_r)} $,
		\begin{align}
			&\big\vert  \langle\Omega, \mathcal{W}_N^*(t;0) \phi_+ ( q_t O \varphi_t) \mathcal{W}_N (t;0) \Omega \rangle\big\vert  \notag \\ 
			\leq&  \frac{C\abs{\l} \| O \|_{\rm op}}{\sqrt{N}}  \|  \varphi_t \|_{\ell^2( B_r)}{ \int_{0}^t ds \;  \exp \bigg( \abs{\l}\int_s^t d\tau \; \| \varphi_\tau \|_{\ell^\infty ( \mathbb{Z}^d)}^2   \bigg)} \langle \Omega, \mathcal{W}^*_N (t;s) ( \mathcal{N} + 1) \mathcal{W}_N (t;s) \Omega \rangle \; .\label{eq:trace-diff-bound2}   \;  
		\end{align} 
		This completes the bound for the second term in the r.h.s.~of \eqref{eq:trace-diff2}.
		
		Lastly, combining \eqref{eq:trace-diff-bound1} and \eqref{eq:trace-diff-bound2} yields Lemma \ref{lemma:trace-diff}. 
		
	\end{proof}
	
	\subsection*{Acknowledgment}	
	The research of M.L. is supported by the Deutsche Forschungsgemeinschaft (DFG, German Research Foundation) through grant TRR 352--470903074. 
	S.R. is supported by the European Research Council via the ERC CoG RAMBAS--Project--Nr. 10104424. 
	J.Z.~is supported by the National Key R \& D Program of China 2022YFA100740, China Postdoctoral Science Foundation Grant 2024T170453, and the Shuimu Scholar program of Tsinghua University. He thanks Jacky Chong and Yulin Gong for very helpful remarks and suggestions. 
	
	\subsection*{Declarations}
	\begin{itemize}
		\item Conflict of interest: The Authors have no conflicts of interest to declare that are relevant to the content of this article.
		\item Data availability: Data sharing is not applicable to this article as no datasets were generated or analysed during the current study.
	\end{itemize}

	\appendix
	\section{Dispersive estimates for discrete Hartree equation}\label{App1}
	In this section we prove the dispersive estimate needed to control various time integrals of the condensate states $\varphi_t$.
	
	We consider discrete Hartree equation  on $\ell^2(\Zb^d),\,d\ge1$ of the form:
	\begin{align}
		\label{Hart}
		i\di_t \varphi = -\Lap \varphi + \l\abs{\varphi}^{2}\varphi.
	\end{align}
	{Here $\abs{\l}\le \l_0$} for some $\l_0>0$  to be determined later. 
	Write $\br{t}:=\sqrt{1+t^2}$ and $\norm{\cdot}_p:= \norm{\cdot}_{\ell^p(\Zb^d)}$.  The key ingredient in proving dispersive estimate for \eqref{Hart} is the $\ell^1\to\ell^\infty$ estimate for the free propagator obtained in Stefanov--Kevrekidis \cite[Thm.~3]{SK}: there exists $C_1=C_1(d)>0$ s.th.~for all $f\in\ell^1(\Zb^d)$, 
	\begin{align}\label{DisEst1}
		\norm{e^{it\Lap}f}_1\le C_1 \br{t}^{-d/3} \norm{f}_\infty.
	\end{align}
	By the $\ell^2$-mass conservation  for \eqref{Hart} and real interpolation , \eqref{DisEst1} implies that for any $2\le p\le\infty$ and $p'=\frac{p}{p-1}$, there exists $C_2=C_2(d,p)>0$ s.th.
	\begin{align}
		\label{DisEst2}
		\norm{e^{it\Lap}f}_{p}\le C_2 \br{t}^{-\al} \norm{f}_{p'},\quad \al:=\frac{d}{3}\frac{p-2}{p}.
	\end{align}
	By a standard fixed-point argument, it follows from \eqref{DisEst2} that
	\begin{corollary}[\cite{KT}, \cite{SK}]\label{corA2}
		Let $d\ge3$. Then for all $2\le q,\,r\le\infty$ satisfying
		\begin{align}
			\label{adm}
			\frac{1}{q}+\frac{d}{3r}\le \frac{d}{6},\quad (q,r,d)\ne(2,\infty,3),
		\end{align}
		there exist $\l_0>0$ and $C>0$ depending only on $d,\,q,\,r$  s.th.~if $\abs{\l}\le \l_0$ in \eqref{Hart} and $\norm{\varphi_0}_2\le1$, then there exists a unique global solution to \eqref{Hart}, $\varphi\in C(\Rb_{\ge0},\ell^q(\Zb^d))$, with initial condition $\varphi_0$, satisfying
		\begin{align}
			\label{phiEst2}
			\norm{\varphi}_{L^q_t\ell^r_x}\le C\norm{\varphi_0}_2.
		\end{align}
	\end{corollary}
	
	For our purpose, we often need to control the integral $\int_0^t \norm{\varphi_\tau}_\infty\,d\tau$ as $t\to\infty$. In dimension $3\le d\le5$, this integral cannot be controlled directly using \eqref{phiEst2} and therefore requires more refined estimate.  
	
	In the next result, we use \eqref{DisEst1} to prove an upper bound on the $\ell^\infty$-decay for solutions to \eqref{NLS} with weak nonlinearity  for $d\ge3$. Related results were obtained in \cite[Thm.~7 ]{SK} for $d=1,2$. See also \cite{HY}.
	
	\begin{theorem}\label{thmA1}
		Let $d\ge3$,  $\varphi_0\in\ell^1(\Zb^d)$, 
		and  $C_1>0$ be the constant from \eqref{DisEst1}. Then there exists $C_2=C_2(d)>0$ s.th.~if  $\norm{\varphi_0}_2\le1$, $\abs{\l}\le \l_0$, and  either
		\begin{align}
			\label{d4}
			\text{ 	$d\ge4$,  $0<\l_0\le \frac12 (C_1C_2) ^{-1}$,   and $t>0$,}
		\end{align}
		or
		\begin{align}
			\label{d3}
			{\text{	 $d=3$, $0<\l_0\le \frac14 (C_1C_2) ^{-2}$,  and $0< t\le T:=e^{1/\sqrt{\abs{\l}}}-1$},}
		\end{align}
		then  the unique solution  $\varphi_t$ to \eqref{Hart} with initial state $\varphi_0$ obeys
		\begin{align}
			\label{phiEst1}
			\norm{\varphi_t}_\infty\le 2C_1\br{t}^{-\al}\norm{\varphi_0}_{1},\quad \al=\frac{d}{3}.
		\end{align}
	\end{theorem}
	\begin{proof}
		Given a solution $\varphi_t,\,t\ge0$ to \eqref{Hart}, by the Duhamel principle, we have
		\begin{align}
			\label{Duh}
			\varphi_t= e^{it\Lap}\varphi_0-i\l\int_0^t e^{i(t-s)\Lap}\abs{\varphi_s}^{2}\varphi_s\,ds,\quad t\ge0.
		\end{align}
		%
		%
		%
		%
		Taking $L^\infty$-norm on both sides of \eqref{Duh} and applying \eqref{DisEst1} to the homogeneous part,	 we obtain
		\begin{align}
			\label{A8}
			\norm{\varphi_t}_\infty\le& C_1 \br{t}^{-\al}\norm{\varphi_0}_{1}+ \abs{\l} \int_0^t   \norm{e^{i(t-s)\Lap}\abs{\varphi_s}^2\varphi_s}_{\infty} \,ds,\quad t\ge0.
		\end{align}
		For $T>0$, define the quantity
		\begin{align}
			\label{}
			M(T):=\sup _{0\le s\le T}\br{s}^\al\norm{{\varphi_s}}_\infty.
		\end{align}
		Our goal is to show that $M(T)$ is a bounded function, globally for $d\ge4$ and up to large $T$ for $d=3$. 
		The remainder proof will be different for $d\ge 4$ and $d=3$, since in the latter case the dispersive estimate \eqref{DisEst1} is too weak to yields the necessary decay.
		
		\textbf{The case \eqref{d4}.}  We claim that the inhomogeneous part in \eqref{A8} is bounded, for $C_1>0$ from \eqref{DisEst1}, some $C_2=C_2(d)>0$ to be specified below, and all $T>0$, as
		\begin{align}
			\label{inhomBdd}
			\int_0^t  \norm{e^{i(t-s)\Lap}\abs{\varphi_s}^{2}\varphi_s}_\infty \,ds\le	 \norm{\varphi_0}_2^2C_1C_2 M(T)\br{t}^{- \al},\quad t\le T. 
		\end{align}
		Indeed, suppose \eqref{inhomBdd} holds.  Then plugging \eqref{inhomBdd} back to \eqref{A8} yields
		$$\br{t}^\al\norm{\varphi_t}_\infty \le C_1\norm{\varphi_0}_1+\abs{\l}\norm{\varphi_0}_2^2C_1C_2M(T),\quad t\le T.$$
		Taking supremum over $t$ on the l.h.s., we arrive at
		\begin{align}
			\label{MTest}
			M(T)\le C_1\norm{\varphi_0}_1 + \abs{\l}\norm{\varphi_0}_2^2C_1C_2 M(T). 
		\end{align} By condition \eqref{d4} and the assumption $\norm{\varphi_0}_2\le1$, we have $\abs{\l}\le  \frac12(\norm{\varphi_0}_2^2C_1C_2)^{-1}$. Plugging this back to \eqref{MTest} and rearranging yields $M(T)\le 2C_1\norm{\varphi_0}_1$. Since this bound holds uniformly for all $T\ge0$, we conclude the desired result \eqref{phiEst1} for the case \eqref{d4} by sending $T\to\infty$.

		Thus it remains to prove \eqref{inhomBdd}. 
		To this end we apply \eqref{DisEst1} to bound the integrand in the l.h.s.~of \eqref{inhomBdd}, yielding
		\begin{align}
			\label{a16}
			\norm{e^{i(t-s)\Lap}\abs{\varphi_s}^2\varphi_s}_{\infty}\le&  \br{t-s}^{-\al} \norm{\abs{\varphi_s}^2\varphi_s}_{1}\notag\\
			\le& \br{t-s}^{-\al}\norm{\varphi_s}_2^2\norm{\varphi_s}_\infty.
		\end{align}
		Using the trivial  $\ell^\infty$-estimate 
		\begin{align}
			\label{infDec}
			\norm{\varphi_s}_\infty\le M(T) \br{s}^{-\al},\quad s\le  T, 
		\end{align}
		the conservation of $\ell^2$-norm, and the fact that $\norm{\varphi}_2\le \norm{\varphi_0}_1$,  inequality \eqref{a16} becomes
		\begin{align}
			\label{a15}
			\norm{e^{i(t-s)\Lap}\abs{\varphi_s}^2\varphi_s}_{\infty}\le\norm{\varphi_0}_2^2M (T)\br{t-s}^{-\al}\br{s}^{-\al},\quad 0\le s\le t\le T.
		\end{align}
		Inserting \eqref{a15} back to \eqref{inhomBdd} yields
		\begin{align}
			\label{inhomBdd1}
			\int_0^t  \norm{e^{i(t-s)\Lap}\abs{\varphi_s}^{2}\varphi_s}_\infty \,ds\le  \norm{\varphi_0}_2^2M (T)\int_0^t\br{t-s}^{-\al}\br{s}^{-\al}\,ds,\quad 0\le t\le T.
		\end{align}
		It remains to prove decay estimate for the integral in the r.h.s.~of \eqref{inhomBdd1}.  We observe that for $0\le s\le t$, either $s>t/2$, or $0\le s\le t/2$ and $t-s\ge t/2$. This implies  the following  identity: for any $\al>0$, $\beta>1$, and $t\ge0$,
		\begin{align}
			\label{convId'}
			\int_{0}^t \br{t-s}^{-\al}\br{s}^{-\beta}\,ds\le&  \br{\tfrac t2}^{-\al}\int_{t/2}^t \br{t-s}^{-\al}\br{s}^{-(\beta-\al)}\,ds+\br{\tfrac t2}^{-\al}\int_{0}^{t/2} \br{s}^{-\beta} \,ds.
		\end{align}
		Consequently, since $\beta>1$, for all $t\ge0$ we have
		\begin{align}\label{convId}
			\int_{0}^t \br{t-s}^{-\al}\br{s}^{-\beta}\,ds\le  & \br{t}^{-\al}C_{\al,\beta},\quad C_{\al,\beta}=2^{1+\al}\int_{0}^\infty\br{s}^{-\beta}\,ds. 
		\end{align}
		{Under condition \eqref{d4}, we have $d\ge4$ and so $d/3>1$ in \eqref{inhomBdd1}.} Therefore the desired estimate \eqref{inhomBdd} follows from \eqref{inhomBdd1}, the bound \eqref{convId} with $\al=\beta=d/3$, and the choice $C_2=C_{d/3,d/3}$. 

		\textbf{The case  \eqref{d3}.} 
		For $d=3$, we have $\al=1$ and therefore proceeding exactly as in the previous case would result in a divergent integral in \eqref{convId}. Thus we instead use \eqref{convId'} to derive the estimate
		\begin{align}
			\label{a20}
			\int_{0}^t \br{t-s}^{-1}\br{s}^{-1}\,ds\le&  4\br{t}^{-1}\int_0^{t/2}\br{s}^{-1}\,ds\le  2\br{\tfrac t2}^{-1}\log (1+T) ,\quad 0\le t\le T.
		\end{align}
		Using \eqref{a20} in place of \eqref{convId} and proceeding as in the derivation of \eqref{MTest}, we arrive at
		\begin{align}
			\label{MTest'}
			M(T)\le C_1\norm{\varphi_0}_1 + \abs{\l}\norm{\varphi_0}_2^2C_1C_2\log(1+T) M(T). 
		\end{align}
		By condition  \eqref{d3} and the assumption $\norm{\varphi_0}_2\le1$, we have  {$\abs{\l}\le \frac12(\norm{\varphi_0}_2^2C_1C_2)^{-1}\sqrt{\abs{\l}}$ and $\sqrt{\abs{\l}}\log(1+T)\le1$.} These, together with \eqref{MTest'}, imply   $M(T)\le 2C_1\norm{\varphi_0}_1$. This establishes \eqref{phiEst1} for the case \eqref{d3} and the proof is complete.

	\end{proof}

	\section{Proof of Lemma \ref{lemma:commutator}} \label{sec:commutator-estimates}

	We prove all three esimates separately and start with the first one.   
	
	\textbf{Proof of \eqref{217}:}  We recall the splitting of the generator  in \eqref{eq:L-sum},
	\begin{align}
		\mathcal{L}_N (t) = \mathbb{H} + \sum_{j=1}^3 \mathcal{R}_{N,t}^{(j)},
	\end{align}
	and compute the commutator of each constituent part  of $\mathcal{L}_N (t)- \dG ( -\Delta)$ with $\sum_{z \in \mathbb{Z}^d} h(z) n_z$ separately. 
	
	We start with the observation that the commutation relations \eqref{eq:CCR} imply 
	\begin{align}\label{218}
		[ \mathbb{H} - \dG ( -\Delta), \sum_{z \in \Zb^d} h(z) a_z^*a_z] =& \frac{\lambda}{2}  \bigg[ \sum_{x \in \Zb^d} \big(  \widetilde{K}_{2,t} (x) b_x^*b_x^* + \overline{\widetilde{K}}_{2,t} (x) b_xb_x \big), \sum_{z \in \Zb^d} h(z) n_z \bigg] .
	\end{align}
	This, together with \eqref{def:b} and  \eqref{eq:CCR}, implies
	\begin{align}
		[ \mathbb{H} - \dG ( -\Delta), \sum_{z \in \Zb^d} h(z) a_z^*a_z] =&  - \lambda \sum_{x \in \Zb^d} \big(  \widetilde{K}_{2,t} (x) h(x) b_x^*b_x^* - \overline{\widetilde{K}}_{2,t} (x) h(x) b_xb_x \big) \; . \label{eq:comm-H}
	\end{align}
	Furthermore, by \eqref{def:b}, we have
	\begin{align}
		\vert b_x b_x \vert =& \sqrt{b_x^*b_x^*b_xb_x} \notag \\
		=& \frac{1}{N}\sqrt{ a_x^* \sqrt{ N - \mathcal{N}^+ (t)}a_x^* \big( N - \mathcal{N}^+ (t) \big) a_x \sqrt{ N - \mathcal{N}^+ (t)} a_x} \notag \\
		\leq& \frac{1}{\sqrt{N}}\sqrt{ a_x^* \sqrt{ N - \mathcal{N}^+ (t)}a_x^*a_x \sqrt{ N - \mathcal{N}^+ (t)} a_x}
	\end{align}
	as an operator inequality on $\mathcal{F}_{\perp \varphi_t}^{\leq N}$. Since $\sqrt{N  - \mathcal{N}^+(t)} a_x^* = a_x^* \sqrt{N  - \mathcal{N}^+(t)-1}$, we obtain 
	\begin{align}
		\vert b_x b_x \vert = \sqrt{b_x^*b_x^*b_xb_x} 
		\leq \frac{1}{\sqrt{N}}\sqrt{ a_x^* a_x^*(N - \mathcal{N}^+ (t)-1) a_x a_x} \leq \sqrt{ a_x^* a_x^* a_x a_x}\; . 
	\end{align}
	Recalling that $a_x^*a_x = n_x$ and $a_x^*n_x a_x = (n_x -1) a_x^*a_x = (n_x -1)n_x$ from \eqref{eq:CCR}, we arrive at 
	\begin{align}
		\vert b_x b_x \vert \leq \sqrt{(n_x-1)n_x} \leq n_x \; . \label{eq:estimate-bb}
	\end{align}
	Proceeding similarly for the hermitian conjugate, we find from \eqref{eq:comm-H} 
	\begin{align}\label{eB7}
		i [ \mathbb{H} - \dG ( -\Delta), \sum_{z \in \Zb^d} h(z) a_z^*a_z] \leq  2 \vert \lambda \vert \sum_{x \in \Zb^d} \vert h(x) \vert \vert \varphi_t (x) \vert^2 \; n_x  
	\end{align}
	where we used that $\vert \widetilde{K}_{2,t}(x) \vert \leq \vert \varphi_t (x) \vert^2$ (see \eqref{tkDef2}). 

	Next we estimate of the remainders $\mathcal{R}_{N,t}^{(j)}$. For this, we note that on $\mathcal{F}_{\perp \varphi_t}^{\leq N}$, the first remainder term defined in \eqref{213} evaluated in a quadratic form on $\mathcal{F}_{\perp \varphi_t}^{\leq N}$ reads 
	\begin{align}
		\mathcal{R}_{N,t}^{(1)} =& \frac{\lambda}{2} \dG ( \vert \varphi_t \vert^2 \varphi_t + \widetilde{K}_{1,t} - \mu_t ) \frac{1-\mathcal{N}_+ (t) }{N} + \lambda \frac{\mathcal{N}^+ (t) }{\sqrt{N}} b(  \vert \varphi_t \vert^2 \varphi_t ) + {\rm h.c.} 
	\end{align}
	On the one hand we have, denoting  $F_t := \dG ( \vert \varphi_t \vert^2 \varphi_t + \widetilde{K}_{1,t} - \mu_t )$, that 
	\begin{align}
		\big[ \sum_{z \in \Zb^d} n_z, \dG ( F_t ) \big] = \sum_{z \in \Zb^d}  \del{ h(z) F_t (z) - F_t (z) h (z)} a_z^*a_z =0. \label{eq:comm-R1-1}
	\end{align}
	On the other hand, with $f_t =  \vert \varphi_t\vert^2 \varphi_t $ , we have
	\begin{align} \label{eq:comm-R1}
		\big[ \sum_{z \in \Zb^d} h(z) n_z, \mathcal{R}_{N,t}^{(1)} \big] =&  \lambda \frac{\mathcal{N}^+ (t)}{\sqrt{N}} b(hf_t) - {\rm h.c.} 
	\end{align}
	These, together with the identity $\mathcal{N}^+ (t) = \sum_{z \in \Zb^d}a^*_za_z$, leads to  
	\begin{align}
		\mathcal{N}^+ (t) \vert b_x \vert = \sum_{z \in \Zb^d} a_z^*a_z \sqrt{b_x^*b_x} = \frac{1}{\sqrt{N}}\sum_{z \in \Zb^d} a_z^* \sqrt{a_x^* (N- \mathcal{N}^+(t)+1 ) a_x } a_z \leq \sum_{z \in \Zb^d} n_z(\sqrt{n_x})  ,
	\end{align}
	and analogously for the hermitian conjugate that leads with $\vert f_t (x) \vert \leq \vert \varphi_t (x) \vert^3$ to
	\begin{align}
		i \big[ \sum_{z \in \Zb^d} h(z) n_z, \mathcal{R}_{N,t}^{(1)} \big] \leq \frac{\vert \lambda \vert }{\sqrt{N}}\sum_{x \in \Zb^d} \vert h(x) \vert \; \vert \varphi (x) \vert^3 \;  {\sqrt{n_x}  \mathcal{N}^+ }\label{225}
	\end{align}
	as an operator inequality on $\mathcal{F}_{\perp \varphi_t}^{\leq N}$. 
	

	Moreover, on $\mathcal{F}_{\perp \varphi_t}^{\leq N}$ we have the identity 
	\begin{align}
		\mathcal{R}_{N,t}^{(2)} = \frac{\lambda}{\sqrt{N}} \sum_{x \in \Zb^d} \varphi_t (x) a^*_xa_xb_x   + {\rm h.c.}
	\end{align}
	leading to 
	\begin{align}
		\big[ \sum_{z \in \Zb^d} h(z) n_z, \mathcal{R}_{N,t}^{(2)} \big] =&  -\frac{\lambda}{\sqrt{N}} \sum_{x \in \Zb^d} \varphi_t (x) h(x) a^*_x a_x b_x - {\rm h.c.} \;  
	\end{align}
Since
	\begin{align}
		\vert a^*_x a_x b_x \vert =& \sqrt{ a_x^*a_x b_x b_x^* a_x^* a_x}\notag \\
		=& \frac{1}{\sqrt{N}}\sqrt{n_x \sqrt{N - \mathcal{N}^+ (t)}(n_x +1 )  \sqrt{N - \mathcal{N}^+ (t)} n_x } \notag \\
		=& \frac{1}{\sqrt{N}}\sqrt{n_x (n_x +1 )^{1/2} \big(N - \mathcal{N}^+ (t) \big)  a_x (n_x +1 )^{1/2}n_x } \notag \\
		\leq& \sqrt{ n_x^2 (n_x +1 )}, \label{eq:estimate-R2}
	\end{align}
	we   conclude that 
	\begin{align}
		i \big[ \sum_{z \in \Zb^d} h(z) n_z, \mathcal{R}_{N,t}^{(2)} \big]  \leq \frac{\vert \lambda \vert }{\sqrt{N}} \sum_{x \in \Zb^d} \vert \varphi_t (x) \vert \; \vert h(x)  \vert n_x(n_x +1 )^{1/2} \; .  \label{eB16}
	\end{align} 
	
	Finally, owning to the identity 
	\begin{align}
		\mathcal{R}_{N,t}^{(3)} = \frac{\lambda}{N} \sum_{x \in \Zb^d} a_x^*a_x^*a_xa_x ,
	\end{align}
	we have 
	\begin{align}
		\big[ \sum_{z \in \Zb^d} h(z) n_z, \mathcal{R}_{N,t}^{(3)} \big] =0 \label{eB18}.
	\end{align}. 

Combining  \eqref{eB7}, \eqref{225}, \eqref{eB16}, and \eqref{eB18}  yields the desired estimate \eqref{217} in Lemma \ref{lemma:commutator}. 
	
	\textbf{Proof of \eqref{eq:comm2}:} 
	%
	%
	%
	%
		Similarly as before, we compute the double commutator of the single contributions of $\mathcal{L}_N (t)$ with $\mathcal{N}$ separately. We start with 
		\begin{align}
			\big[ \mathcal{N}, \big[ \mathcal{N},  \mathbb{H} \big] \big] = 2 \lambda \sum_{x \in \Zb^d} \big( \widetilde{K}_{2,t} (x) b_x^*b_x^* + \overline{\widetilde{K}}_{2,t} (x) b_x b_x \big) 
		\end{align}
		where we used that $\big[\dG ( -\Delta), \mathcal{N} \big] =0$. Since $\big( \mathcal{N} + 3 \big)^{-1/2} b_x b_x =  b_x b_x \big( \mathcal{N} +1  \big)^{-1/2}$ and $\vert b_x b_x \vert \leq n_x$ from \eqref{eq:estimate-bb}, we have on the one hand 
		\begin{align}
			\big\| \big( \mathcal{N} + 3\big)^{-1/2} \sum_{x \in \Zb^d} \overline{\widetilde{K}}_{2,t} (x) b_x b_x  \psi \big\| \leq  \| \varphi_t \|_{\ell^\infty}^2 \| \big( \mathcal{N} + 1)^{1/2} \psi \|  \label{eq:comm-2-K2}
		\end{align}
		where we used that $\vert \widetilde{K}_{2,t} ( x) \vert \leq \vert \varphi_t (x) \vert^2$. On the other hand with $\big( \mathcal{N} + 3 \big)^{1/2} b_x^* b_x^* =  b_x^* b_x^* \big( \mathcal{N} +5  \big)^{1/2}$ we find for any $\xi,\psi \in \mathcal{F}_{\perp \varphi_t}^{\leq N}$ that
		\begin{align}
			\langle \xi, &  \big( \mathcal{N} + 3\big)^{-1/2} \sum_{x \in \Zb^d} \widetilde{K}_{2,t} (x) b_x^* b_x^*  \psi \rangle \leq \|  \sum_{x \in \Zb^d} \overline{\widetilde{K}}_{2,t} (x) b_x b_x ( \mathcal{N} + 3)^{-1} \xi \| \; \| ( \mathcal{N} + 5)^{1/2} \psi \|   \notag \\ 
			\leq&  \| \varphi_t \|_{\ell^\infty}^2 \| \xi \| \;  \| (\mathcal{N} + 5)^{1/2} \psi \|
		\end{align}
		where we used \eqref{eq:comm-2-K2}. With $\| ( \mathcal{N} + 5 )^{1/2} \psi \| \leq C \| ( \mathcal{N} + 1)^{1/2} \psi \|$ we thus arrive at 
		\begin{align}
			\big\vert \big\langle \xi,  ( \mathcal{N} + 3)^{-1/2} \big[ \mathcal{N}, \big[ \mathcal{N}, \mathbb{H} - \dG ( - \Delta ) \big] \big] \psi \big\rangle \big\vert \leq 20  \vert \lambda \vert  \| \varphi_t \|_{\ell^\infty}^2 \| ( \mathcal{N} + 1)^{1/2} \psi \| \;  \| \xi \| \; . 
		\end{align}
		
		To estimate the double commutator of $\mathcal{R}_{N,t}^{(1)}$ with $\mathcal{N}$, we find with \eqref{eq:comm-R1-1} and \eqref{eq:comm-R1} that
		\begin{align}
			\big[ \mathcal{N}, \big[ \mathcal{N}, \mathcal{R}_{N,t}^{(1)} \big] \big] =&  \lambda \frac{\mathcal{N}^+ (t)}{\sqrt{N}} b( f_t ) + \lambda  b^*( f_t ) \frac{\mathcal{N}^+ (t)}{\sqrt{N}}
		\end{align}
		where we used the notation $f_t =  \vert \varphi_t \vert^2 \varphi_t$. Since $ ( \mathcal{N} + 3)^{-1/2} b (f_t) =  b(f_t) ( \mathcal{N} + 2)^{-1/2}$ resp. $( \mathcal{N} + 3)^{-1/2} b^*(f_t)  = b^*(f_t) ( \mathcal{N} + 4)^{-1/2}$,  we find by Cauchy Schwarz for any $\psi \in \mathcal{F}_{\perp \varphi_t}^{\leq N}$ that
		\begin{align}
			\| ( \mathcal{N} + 3)^{-1/2} \big[ \mathcal{N}, \big[ \mathcal{N}, \mathcal{R}_{N,t}^{(1)} \big] \big] \psi \| \leq 4  \vert \lambda\vert \| f_t \|_{\ell^2} \| ( \mathcal{N} + 1)^{1/2} \psi \| 
			\leq 4 \vert \lambda \vert  \| \varphi_t \|_{\ell^\infty}^2  \| ( \mathcal{N} + 1)^{1/2} \psi \|
		\end{align}
		where we used that $\| f_t \|_{\ell^2} \leq \| \varphi_t \|_{\ell^\infty}^2 \| \varphi_t \|_{\ell^2} = \| \varphi_t \|_{\ell^\infty}^2 $. 
		
		The double commutator of $\mathcal{R}_{N,t}^{(2)}$ with $\mathcal{N}$ is given by 
		\begin{align}
			\big[ \mathcal{N}, \big[ \mathcal{N}, \mathcal{R}_{N,t}^{(2)}\big] \big]  = \frac{\lambda}{\sqrt{N}}\sum_{x \in \Zb^d} \big(  \varphi_t (x) a_x^*a_x b_x + \overline{\varphi}_t (x) b_x^*a^*_x a_x \big) \; . 
		\end{align} 
		On the one hand we have by Cauchy Schwarz 
		\begin{align}
			\big\vert \big\langle \xi,  & ( \mathcal{N} + 3)^{-1/2}\frac{1}{\sqrt{N}}\sum_{x \in \Zb^d}  \varphi_t (x) a_x^*a_x b_x \psi \big\rangle \big\vert \notag \\
			\leq& \frac{1}{\sqrt{N}} \bigg( \sum_{x\in \Zb^d} \| a_x^*a_x ( \mathcal{N} + 3)^{-1/2}\xi \|^2 \big)^{1/2} \big( \sum_{x \in \Zb^d} \vert \varphi_t (x) \vert^2 \; \| b_x \psi \|^2 \big)^{1/2} \notag \\
			\leq&  \| \varphi_t \|_{\ell^\infty} \| ( \mathcal{N} + 1 )^{1/2} \psi \|\; \| \xi \| 
		\end{align},
		and on the other hand, with similar arguments 
		\begin{align}
			\big\vert \big\langle \xi,  & ( \mathcal{N} + 3)^{-1/2}\frac{1}{\sqrt{N}}\sum_{x \in \Zb^d}  \overline{\varphi}_t (x) b_x^*a_x^* a_x \psi \big\rangle \big\vert \notag \\
			\leq& \frac{1}{\sqrt{N}} \bigg( \sum_{x\in \Zb^d} \| a_x^*a_x \psi \|^2 \big)^{1/2} \big( \sum_{x \in \Zb^d} \vert \varphi_t (x) \vert^2 \; \| b_x ( \mathcal{N} + 3)^{-1/2}\xi \|^2 \big)^{1/2} \notag \\
			\leq&  \| \varphi_t \|_{\ell^\infty} \| ( \mathcal{N} + 1 )^{1/2} \psi \|  \; \| \xi \|  
		\end{align},
		so that we get 
		\begin{align}
			\big\vert \big\langle \xi, ( \mathcal{N} + 3)^{-1/2} \big[ \mathcal{N}, \big[ \mathcal{N}, \mathcal{R}_{N,t}^{(2)}\big] \big]  \psi \big\rangle \big\vert \leq  \vert \lambda \vert \;  \| \varphi_t \|_{\ell^\infty} \| ( \mathcal{N} + 1 )^{1/2} \psi \|  \; \| \xi \|   \; .
		\end{align}
		With the observation $[\mathcal{N}, \mathcal{R}_{N,t}^{(3)}] =0$ this finally leads to \eqref{eq:comm2}.
		
		\textbf{Proof of \eqref{eq:comm1}:} As before, we estimate the single contribution of $\mathcal{L}_N (t)$ 
		separately. We start with 
		\begin{align}
			\big[ \mathcal{N},  \mathbb{H} \big]  =  \lambda \sum_{x \in \Zb^d} \big( \widetilde{K}_{2,t} (x) b_x^*b_x^* - \overline{\widetilde{K}}_{2,t} (x) b_x b_x \big)  \; . 
		\end{align}
		As in the proof of \eqref{eq:comm2}, we have on the one hand 
		\begin{align}
			\big\| \sum_{x \in \Zb^d} \overline{\widetilde{K}}_{2,t} (x) b_x b_x  \psi \big\| \leq  \| \varphi_t \|_{\ell^\infty}^2 \| \big( \mathcal{N} + 1) \psi \|  .
		\end{align}
		On the other hand with $ (\mathcal{N}  +1)b_x^* b_x^* =  b_x^* b_x^* \big( \mathcal{N} +3  \big) $ we find for any $\xi,\psi \in \mathcal{F}_{\perp \varphi_t}^{\leq N}$
		\begin{align}
			\langle \xi, &  \sum_{x \in \Zb^d} \widetilde{K}_{2,t} (x) b_x^* b_x^*  \psi \rangle \leq \|  \sum_{x \in \Zb^d} \overline{\widetilde{K}}_{2,t} (x) b_x b_x ( \mathcal{N} + 1)^{-1}\xi \| \; \| ( \mathcal{N} + 3)^{1} \psi \|   \notag \\ 
			\leq&  \| \varphi_t \|_{\ell^\infty}^2 \| \xi \| \;  \| (\mathcal{N} + 3) \psi \| \; . 
		\end{align}
		The remaining terms of the commutator $[\mathcal{L}_N (t), \mathcal{N}]$ can be estimated similarly along the lines of the proof of \eqref{eq:comm2} before, choosing different weights of $( \mathcal{N} + 1)^j$ as in the previous estimates.


		\section{Proof of \lemref{lem23}}\label{secPfLem23}
		
		By Gr\"onwall's inequality, it suffices to prove that there exists some universal constant $C_*>0$ s.th.~for $j=1,2,3$, 
		\begin{align} 
			{\od{}{t}	\big\langle ( \mathcal{N} + 1)^j \big\rangle_{(t;s)}} \leq C_* \abs{\lambda}\norm{\varphi_t}_{\ell^\infty} \langle (\mathcal{N} + 1 )^j \rangle_{(t;s)},\quad t\ge s. \label{252}
		\end{align} 
		We compute, using that $[\dG(-\Lap),\cN]=\dG([ -\Lap ,1])=0$,
		\begin{align}\label{233}
			\frac{d}{dt} \langle ( \mathcal{N} + 1)^j \rangle_{(t;s)} =  \langle i\big[ \mathcal{L}_N(t) - \dG ( -\Delta), (\mathcal{N}+1)^j  \big] \rangle_{(t;s)}.
		\end{align}
		Therefore \eqref{252} follow as long as 
		\begin{align}
			\label{252'}
			\langle \xi,i\big[ \mathcal{L}_N(t) - \dG ( -\Delta), (\mathcal{N}+1)^j \xi \big] \rangle\le  C_* \abs{\lambda}\norm{\varphi_t}_{\ell^\infty} \langle\xi, (\mathcal{N} + 1 )^j \xi\rangle,\quad \xi\in\Fl.
		\end{align} 
		
		\textbf{Proof of \eqref{252'} $j=1$:} 
		By \eqref{217} for $h(x)\equiv1$, we have
		\begin{align}\label{217'}
			i\big[ \mathcal{L}_N(t) - \dG ( -\Delta), \mathcal{N}  \big]
			\leq& 2 \vert \lambda \vert  \sum_{x \in \Zb^d}   \bigg( 2 \vert \varphi_t (x) \vert^2 n_x  + \frac{\vert \varphi_t (x) \vert^3}{\sqrt{N}} \mathcal{N} n_x^{1/2} + \frac{\vert \varphi_t (x) \vert}{\sqrt{N}} n_x (n_x+1)^{1/2} \bigg).
		\end{align}
		Through a string of elementary inequalities, the terms in the r.h.s.~of \eqref{217'} can be bounded as follows:
		\begin{align}
			\sum_{x \in \Zb^d}    \vert \varphi_t (x) \vert^2 n_x   \le & { \norm{\varphi_t}_{\ell^\infty}^2 }\cN,\\
			{ \sum_{x \in \Zb^d}   	\frac{\vert \varphi_t (x) \vert^3}{\sqrt{N}} \mathcal{N} n_x^{1/2}} 
			\le &   \frac\cN{\sqrt{N}} \del{\sum \abs{\varphi_t(x)}^6} ^{1/2} \del{\sum  n_x }^{1/2}\notag\\&\le \norm{\varphi_t}_{\ell^\infty}^2\norm{\varphi_t}_{\ell^2}\frac{\cN^{3/2}}{\sqrt{N}},\\
			2 { \sum_{x \in \Zb^d}  	 \frac{\vert \varphi_t (x) \vert}{\sqrt{N}} n_x (n_x+1)^{1/2}}
			\le& \sum \abs{\varphi_t(x)}  \del{ n_x+\frac{1}{N}    {n_x(n_x+1)}}\notag\\& 	
			\le{\norm{\varphi_t}_{\ell^\infty} }\del{\cN+\frac{\cN+\cN^2}{N} }. 		
		\end{align}
		Combining the estimates above and using that $\norm{\varphi_t}_{\ell^\infty}\le \norm{\varphi_t}_\ell^2= \norm{\varphi_t}_{\ell^2}\le 1$, we find that {for any $\xi\in\Fl$,}
		\begin{align}
			\label{413}
			\br{\xi,i\big[ \mathcal{L}_N(t) - \dG ( -\Delta), \mathcal{N}  \big]\xi}\le 8\abs{\l}\norm{\varphi_t}_{\ell^\infty} \br{\xi,(\cN+1)\xi}.
		\end{align}
		This yields \eqref{252'} for $j=1$.

		\textbf{Proof of\eqref{252'} for $j=2$:} We compute, by the Leibniz rule for commutators,
		\begin{align}
			& \big[ \mathcal{L}_N (t) - \dG (- \Delta), (\mathcal{N}+1)^2 \big]\notag\\ =&  (\mathcal{N}+1) \big[ \mathcal{L}_N (t)  - \dG (-\Delta), \mathcal{N}\big]   +  \big[ \mathcal{L}_N (t)  - \dG (-\Delta), \mathcal{N}\big] (\mathcal{N}+1)   \notag \\
			=& 2    (\mathcal{N}+1) \big[ \mathcal{L}_N (t)  - \dG (-\Delta), \mathcal{N}\big]   +  \big[ \big[ \mathcal{L}_N (t)  - \dG (-\Delta), \mathcal{N}\big], \mathcal{N}  \big]   \; . \label{101}
		\end{align}
		To bound the first term, we use \eqref{eq:comm1} to get
		\begin{align}
			\label{}
			&\br{\xi, (\mathcal{N}+1) \big[ \mathcal{L}_N (t)  - \dG (-\Delta), \mathcal{N}\big]\xi}\notag\\=&\br{(\mathcal{N}+1)\xi,  \big[ \mathcal{L}_N (t)  - \dG (-\Delta), \mathcal{N}\big]\xi}\notag\\
			\le& C \vert \lambda \vert  \| \varphi_t \|_{\ell^\infty} \norm{(\cN+1)\xi}^2.
		\end{align}
		To bound the second term in the r.h.s.~of \eqref{101}, we use the double commutator bound \eqref{eq:comm2} to obtain
		\begin{align}
			&\xidel{\big[ \big[ \mathcal{L}_N (t)  - \dG (-\Delta), \mathcal{N}\big], \mathcal{N}  \big]}\notag\\=&\br{(\mathcal{N}+3)^{1/2}\xi,  (\mathcal{N}+3)^{-1/2}\big[ \big[ \mathcal{L}_N (t)  - \dG (-\Delta), \mathcal{N}\big], \mathcal{N}  \big]\xi}\notag\\
			\le& C \vert \lambda \vert  \| \varphi_t \|_{\ell^\infty}  \norm{(\cN+3)^{1/2}\xi}\norm{(\cN+1)^{1/2}\xi}. 
		\end{align}
		Since $\norm{(\cN+1)\xi}^2=\xidel{(\cN+1)^2}$ and $\norm{(\cN+3)^{1/2}\xi}\norm{(\cN+1)^{1/2}\xi}\le C\xidel{(\cN+1)}$, we conclude \eqref{252'} for $j=2$. 
		
		\textbf{Proof of \eqref{252'} for $j=3$:} As before, we compute   with the Leibniz rule 
		\begin{align}
			\big[ (\mathcal{N} + 1)^3, \mathcal{L}_N (t)  - \dG ( -\Delta) \big] =& 3  ( \mathcal{N} + 1) \big[  \mathcal{N}, \mathcal{L}_N (t)  - \dG ( -\Delta) \big]( \mathcal{N} + 1) \notag \\
			&+ ( \mathcal{N} + 1) \big[ \mathcal{N}, \big[  \mathcal{N}, \mathcal{L}_N (t)  - \dG ( -\Delta) \big]\big] \notag \\
			&+ {\big[ \big[  \mathcal{N}, \mathcal{L}_N (t)  - \dG ( -\Delta) \big], \mathcal{N} \big] (\mathcal{N} +1)} \; . 
		\end{align}
		Using   \eqref{413} for the first term on the r.h.s.~and the bound  \eqref{eq:comm2} for the second and third, we find
		\begin{align}
			&\xidel{ \big[ (\mathcal{N} + 1)^3, \mathcal{L}_N (t)  - \dG ( -\Delta) \big]}\notag \\
			&\leq C \vert \lambda \vert \;  \| \varphi_t \|_{\ell^\infty}  \del{\xidel{( \mathcal{N} + 1)^3} +    \xidel{( \mathcal{N} + 1)}  ^{1/2}    \xidel{( \mathcal{N} + 1)^{3} } ^{1/2}}  \notag \\
			&\leq C \vert \lambda \vert \; \| \varphi_t \|_{\ell^\infty}  \xidel{( \mathcal{N} + 1)^3} \; . 
		\end{align}
		This yields \eqref{252'} for $j=3$.
		
		\nocite{}
		\bibliographystyle{alpha}
		\bibliography{LocFlucEst}

\begin{thebibliography}{SHOE11}

\bibitem[AFPS23]{arbunich2023maximal}
Jack Arbunich, J{\'e}r{\'e}my Faupin, Fabio Pusateri, and Israel~Michael Sigal.
\newblock {Maximal speed of quantum propagation for the Hartree equation}.
\newblock {\em Communications in Partial Differential Equations},
  48(4):542--575, 2023.

\bibitem[AKS13]{AKS}
G{\'e}rard~Ben Arous, Kay Kirkpatrick, and Benjamin Schlein.
\newblock A central limit theorem in many-body quantum dynamics.
\newblock {\em Communications in Mathematical Physics}, 321(2):371--417, May
  2013.

\bibitem[BdOS15]{benedikter2015quantitative}
Niels Benedikter, Gustavo de~Oliveira, and Benjamin Schlein.
\newblock {Quantitative Derivation of the Gross-Pitaevskii Equation}.
\newblock {\em Communications on Pure and Applied Mathematics},
  68(8):1399--1482, 2015.

\bibitem[BPPS22]{bossmann2022beyond}
Lea Bo{\ss}mann, S{\"o}ren Petrat, Peter Pickl, and Avy Soffer.
\newblock Beyond {B}ogoliubov dynamics.
\newblock {\em Pure and Applied Analysis}, 3(4):677--726, 2022.

\bibitem[BPS14]{benedikter2014mean}
Niels Benedikter, Marcello Porta, and Benjamin Schlein.
\newblock Mean--field evolution of fermionic systems.
\newblock {\em Communications in Mathematical Physics}, 331:1087--1131, 2014.

\bibitem[BPS16]{BPS}
Niels Benedikter, Marcello Porta, and Benjamin Schlein.
\newblock {\em Effective Evolution Equations from Quantum Dynamics}.
\newblock Springer International Publishing, 2016.

\bibitem[BS19]{brennecke2019gross}
Christian Brennecke and Benjamin Schlein.
\newblock {Gross--Pitaevskii dynamics for Bose--Einstein condensates}.
\newblock {\em Analysis \& PDE}, 12(6):1513--1596, 2019.

\bibitem[CH16]{chen2016klainerman}
Xuwen Chen and Justin Holmer.
\newblock On the {K}lainerman--{M}achedon conjecture for the quantum {BBGKY}
  hierarchy with self-interaction.
\newblock {\em Journal of the European Mathematical Society}, 18(6):1161--1200,
  2016.

\bibitem[CH19]{chen2019derivation}
Xuwen Chen and Justin Holmer.
\newblock The derivation of the $\mathbb{T}^3$ energy-critical {NLS} from
  quantum many-body dynamics.
\newblock {\em Inventiones mathematicae}, 217:433--547, 2019.

\bibitem[DCS23]{deuchert2023dynamics}
Andreas Deuchert, Marco Caporaletti, and Benjamin Schlein.
\newblock Dynamics of mean-field bosons at positive temperature.
\newblock {\em Annales de l'Institut Henri Poincar\'{e} (C) Analyse Non
  Lin\'{e}aire}, 4(41):995--1054, 2023.

\bibitem[DL23]{DL}
Charlotte Dietze and Jinyeop Lee.
\newblock {Uniform in Time Convergence to Bose-Einstein Condensation for a
  Weakly Interacting Bose Gas with an External Potential}.
\newblock In {\em Quantum Mathematics II}, pages 267--311. Springer Nature
  Singapore, 2023.

\bibitem[EESY04]{elgart2004nonlinear}
Alexander Elgart, L{\'a}szl{\'o} Erd{\H{o}}s, Benjamin Schlein, and Horng-Tzer
  Yau.
\newblock Nonlinear {H}artree equation as the mean field limit of weakly
  coupled fermions.
\newblock {\em Journal de math{\'e}matiques pures et appliqu{\'e}es},
  83(10):1241--1273, 2004.

\bibitem[ES09]{erdHos2009quantum}
L{\'a}szl{\'o} Erd{\H{o}}s and Benjamin Schlein.
\newblock Quantum dynamics with mean field interactions: a new approach.
\newblock {\em Journal of Statistical Physics}, 134(5):859--870, 2009.

\bibitem[ESY09]{erdHos2009rigorous}
L{\'a}szl{\'o} Erd{\H{o}}s, Benjamin Schlein, and Horng-Tzer Yau.
\newblock {Rigorous derivation of the Gross-Pitaevskii equation with a large
  interaction potential}.
\newblock {\em Journal of the American Mathematical Society}, 22(4):1099--1156,
  2009.

\bibitem[ESY10]{erdos2010derivation}
L{\'a}szl{\'o} Erdos, Benjamin Schlein, and Horng-Tzer Yau.
\newblock {Derivation of the Gross-Pitaevskii equation for the dynamics of
  Bose-Einstein condensate}.
\newblock {\em Annals of Mathematics}, 172:291--370, 2010.

\bibitem[FLS22a]{FLSa}
J\'{e}r\'{e}my Faupin, Marius Lemm, and Israel~Michael Sigal.
\newblock Maximal speed for macroscopic particle transport in the
  {B}ose-{H}ubbard model.
\newblock {\em Phys. Rev. Lett.}, 128(15):Paper No. 150602, 6, 2022.

\bibitem[FLS22b]{FLS}
J\'{e}r\'{e}my Faupin, Marius Lemm, and Israel~Michael Sigal.
\newblock On {L}ieb-{R}obinson bounds for the {B}ose-{H}ubbard model.
\newblock {\em Commun. Math. Phys.}, 394(3):1011--1037, 2022.

\bibitem[FPS23]{fresta2023effective}
Luca Fresta, Marcello Porta, and Benjamin Schlein.
\newblock Effective dynamics of extended fermi gases in the high-density
  regime.
\newblock {\em Communications in Mathematical Physics}, 401(2):1701--1751,
  2023.

\bibitem[GMM09]{GMM}
Manoussos~G. Grillakis, Matei Machedon, and Dionisios Margetis.
\newblock Second-order corrections to mean field evolution of weakly
  interacting bosons. {I}.
\newblock {\em Communications in Mathematical Physics}, 294(1):273--301, 2009.

\bibitem[GMM11]{GMMa}
M.~Grillakis, M.~Machedon, and D.~Margetis.
\newblock Second-order corrections to mean field evolution of weakly
  interacting bosons. {II}.
\newblock {\em Advances in Mathematics}, 228(3):1788--1815, October 2011.

\bibitem[HS21]{huang2021uncertainty}
Shanlin Huang and Avy Soffer.
\newblock Uncertainty principle, minimal escape velocities, and observability
  inequalities for {S}chr{\"o}dinger equations.
\newblock {\em American Journal of Mathematics}, 143(3):753--781, 2021.

\bibitem[HY19]{HY}
Younghun Hong and Changhun Yang.
\newblock Uniform {S}trichartz estimates on the lattice.
\newblock {\em Discrete and Continuous Dynamical Systems - A},
  39(6):3239--3264, 2019.

\bibitem[KL24]{kuwahara2024enhanced}
Tomotaka Kuwahara and Marius Lemm.
\newblock Enhanced {L}ieb-{R}obinson bounds for a class of {B}ose-{H}ubbard
  type hamiltonians.
\newblock {\em arXiv preprint, arXiv:2405.04672}, 2024.

\bibitem[KP10]{knowles2010mean}
Antti Knowles and Peter Pickl.
\newblock Mean-field dynamics: singular potentials and rate of convergence.
\newblock {\em Communications in Mathematical Physics}, 298:101--138, 2010.

\bibitem[KS21]{kuwahara2021lieb}
Tomotaka Kuwahara and Keiji Saito.
\newblock Lieb-{R}obinson bound and almost-linear light cone in interacting
  boson systems.
\newblock {\em Physical Review Letters}, 127(7):070403, 2021.

\bibitem[KT98]{KT}
Markus Keel and Terence Tao.
\newblock Endpoint {S}trichartz estimates.
\newblock {\em American Journal of Mathematics}, 120(5):955--980, October 1998.

\bibitem[Kuz15]{Kuz}
Elif Kuz.
\newblock Rate of convergence to mean field for interacting bosons.
\newblock {\em Communications in Partial Differential Equations},
  40(10):1831--1854, June 2015.

\bibitem[KVS24]{kuwahara2024effective}
Tomotaka Kuwahara, Tan~Van Vu, and Keiji Saito.
\newblock Effective light cone and digital quantum simulation of interacting
  bosons.
\newblock {\em Nature Communications}, 15(1):2520, 2024.

\bibitem[Lee19]{Lee}
Jinyeop Lee.
\newblock {On the Time Dependence of the Rate of Convergence Towards Hartree
  Dynamics for Interacting Bosons}.
\newblock {\em Journal of Statistical Physics}, 176(2):358--381, May 2019.

\bibitem[LNS15]{lewin2015fluctuations}
Mathieu Lewin, Phan~Th{\`a}nh Nam, and Benjamin Schlein.
\newblock Fluctuations around {H}artree states in the mean-field regime.
\newblock {\em American Journal of Mathematics}, 137(6):1613--1650, 2015.

\bibitem[LNSS15]{LNSS15}
Mathieu Lewin, Phan~Thanh Nam, Sylvia Serfaty, and Jan~Philip Solovej.
\newblock {B}ogoliubov spectrum of interacting {B}ose gases.
\newblock {\em Communications on Pure and Applied Mathematics}, 68(3):413--471,
  2015.

\bibitem[LR72]{lieb1972finite}
Elliott~H Lieb and Derek~W Robinson.
\newblock The finite group velocity of quantum spin systems.
\newblock {\em Communications in mathematical physics}, 28(3):251--257, 1972.

\bibitem[LR23]{RL23}
Marius Lemm and Simone Rademacher.
\newblock {Out-of-time-ordered correlators of mean-field bosons via Bogoliubov
  theory}.
\newblock {\em arXiv preprint, arXiv:2312.01736}, 2023.

\bibitem[LRSZ23]{LRSZ}
Marius Lemm, Carla Rubiliani, Israel~Michael Sigal, and Jingxuan Zhang.
\newblock Information propagation in long-range quantum many-body systems.
\newblock {\em Phys. Rev. A}, 108:L060401, Dec 2023.

\bibitem[LRZ23]{LRZ}
Marius Lemm, Carla Rubiliani, and Jingxuan Zhang.
\newblock On the microscopic propagation speed of long-range quantum many-body
  systems.
\newblock {\em arXiv preprint, arXiv:2310.14896}, 2023.

\bibitem[LS23]{lafleche2023strong}
Laurent Lafleche and Chiara Saffirio.
\newblock Strong semiclassical limits from {H}artree and {H}artree--{F}ock to
  {V}lasov--{P}oisson equations.
\newblock {\em Analysis \& PDE}, 16(4):891--926, 2023.

\bibitem[MPP19]{mitrouskas2019bogoliubov}
David Mitrouskas, S{\"o}ren Petrat, and Peter Pickl.
\newblock Bogoliubov corrections and trace norm convergence for the {H}artree
  dynamics.
\newblock {\em Reviews in Mathematical Physics}, 31(08):1950024, 2019.

\bibitem[Nap23]{napiorkowski2023dynamics}
Marcin Napi{\'o}rkowski.
\newblock Dynamics of interacting bosons: a compact review.
\newblock In {\em Density Functionals For Many-particle Systems: Mathematical
  Theory And Physical Applications Of Effective Equations}, pages 117--154.
  World Scientific, 2023.

\bibitem[NN17]{nam2017bogoliubov}
Phan~Th{\`a}nh Nam and Marcin Napi{\'o}rkowski.
\newblock Bogoliubov correction to the mean-field dynamics of interacting
  bosons.
\newblock {\em Advances in Theoretical and Mathematical Physics},
  21(3):683--738, 2017.

\bibitem[NRSS07]{nachtergaele2007lieb}
Bruno Nachtergaele, Hillel Raz, Benjamin Schlein, and Robert Sims.
\newblock Lieb-{R}obinson bounds for harmonic and anharmonic lattice systems.
\newblock {\em arXiv preprint, arXiv:0712.3820}, 2007.

\bibitem[Pic11]{pickl2011simple}
Peter Pickl.
\newblock A simple derivation of mean field limits for quantum systems.
\newblock {\em Letters in Mathematical Physics}, 97:151--164, 2011.

\bibitem[RS09]{RS}
Igor Rodnianski and Benjamin Schlein.
\newblock Quantum fluctuations and rate of convergence towards mean field
  dynamics.
\newblock {\em Communications in Mathematical Physics}, 291(1):31--61, July
  2009.

\bibitem[RS22]{RS22}
Simone Rademacher and Robert Seiringer.
\newblock {L}arge deviation estimates for weakly interacting bosons.
\newblock {\em Journal of Statistical Physics}, 188(1):9, 2022.

\bibitem[SHOE11]{schuch2011information}
Norbert Schuch, Sarah~K Harrison, Tobias~J Osborne, and Jens Eisert.
\newblock {Information propagation for interacting-particle systems}.
\newblock {\em Physical Review A}, 84(3):032309, 2011.

\bibitem[SK05]{SK}
Atanas Stefanov and Panayotis~G Kevrekidis.
\newblock {Asymptotic behaviour of small solutions for the discrete nonlinear
  Schr\"odinger and Klein-Gordon equations}.
\newblock {\em Nonlinearity}, 18(4):1841--1857, May 2005.

\bibitem[YL22]{yin2022finite}
Chao Yin and Andrew Lucas.
\newblock Finite speed of quantum information in models of interacting bosons
  at finite density.
\newblock {\em Physical Review X}, 12(2):021039, 2022.

\bibitem[Zha24a]{Zhaa}
Jingxuan Zhang.
\newblock Spectral localization estimates for abstract linear {S}chr\"odinger
  equations.
\newblock {\em arXiv preprint, arXiv:2409.10873}, 2024.

\bibitem[Zha24b]{Zha}
Jingxuan Zhang.
\newblock Upper bounds in non-autonomous quantum dynamics.
\newblock {\em arXiv preprint, arXiv:2409.13762}, 2024.

\end{thebibliography}
	\end{document}